\documentclass[conference]{IEEEtran}

\usepackage{etoolbox}
\usepackage{enumitem}
\usepackage[dvipsnames,svgnames]{xcolor}

\usepackage[noadjust]{cite}

\usepackage{tikz}
\usetikzlibrary{arrows,arrows.meta,shapes.geometric,math}

\usepackage{pgfplots}
\pgfplotsset{compat=1.15}

\usepackage{mathtools}
\usepackage{nccmath}
\usepackage{amssymb}
\usepackage{mathdots}
\usepackage{upgreek}

\usepackage[caption=false]{subfig}
\usepackage{commath}
\usepackage{amsthm}
\usepackage{graphicx}
\usepackage{pgfgantt}
\usepackage{xfrac}
\usepackage{mathrsfs}
\usepackage{bbm}
\usepackage[draft=false]{hyperref}

\usepackage{cleveref}
\crefname{equation}{Eq}{} 
\crefrangelabelformat{equation}{(#3#1#4--#5#2#6)}

\newcommand{\floor}[1]{\left \lfloor #1 \right \rfloor}

\usepackage{autobreak}
\allowdisplaybreaks

\def \D{\mathcal{D}}
\def \E{\mathcal{E}}

\def \N{\mathcal{N}}

\def \S{\mathcal{S}}

\def \fB{\mathbf{B}}
\def \fA{\mathbf{A}}

\def \fc{\mathbf{c}}
\def \fx{\mathbf{x}}

\def \bG {\mathcal{G}_{\fh}}

\def \fh{\mathbf{h}}
\def \fH{\mathbf{H}}

\def \fv{\mathbf{v}}

\def \fx{\mathbf{x}}
\def \fy{\mathbf{y}}
\def \fz{\mathbf{z}}

\def \fY{\mathbf{Y}}
\def \fZ{\mathbf{Z}}
\def \f0{\mathbf{0}}

\definecolor{blau_1a}{RGB}{93,133,195}
\definecolor{blau_2a}{RGB}{0,156,218}
\definecolor{gruen_3a}{RGB}{80,182,149}
\definecolor{gruen_4a}{RGB}{175,204,80}
\definecolor{gruen_5a}{RGB}{221,223,72}
\definecolor{orange_6a}{RGB}{255,224,92}
\definecolor{orange_7a}{RGB}{248,186,60}
\definecolor{rot_8a}{RGB}{238,122,52}
\definecolor{rot_9a}{RGB}{233,80,62}
\definecolor{lila_10a}{RGB}{201,48,142}
\definecolor{lila_11a}{RGB}{128,69,151}

\definecolor{blau_1b}{RGB}{0,90,169}
\definecolor{blau_2b}{RGB}{0,131,204}
\definecolor{gruen_3b}{RGB}{0,157,129}
\definecolor{gruen_4b}{RGB}{153,192,0}
\definecolor{gruen_5b}{RGB}{201,212,0}
\definecolor{orange_6b}{RGB}{253,202,0}
\definecolor{orange_7b}{RGB}{245,163,0}
\definecolor{rot_8b}{RGB}{236,101,0}
\definecolor{rot_9b}{RGB}{230,0,26}
\definecolor{lila_10b}{RGB}{166,0,132}
\definecolor{lila_11b}{RGB}{114,16,133}

\definecolor{mycolor1}{rgb}{0.0, 0.18, 0.39}
\definecolor{mycolor2}{RGB}{87,108,67}
\definecolor{mycolor3}{RGB}{8,133,161}
\definecolor{mycolor4}{RGB}{80,91,161}
\definecolor{mycolor5}{RGB}{98,122,157}
\definecolor{mycolor6}{RGB}{255,163,67}
\definecolor{mycolor7}{RGB}{152,205,225}
\definecolor{mycolor8}{RGB}{242,204,48}
\definecolor{mycolor9}{rgb}{0,.5,0}
\definecolor{mycolor10}{rgb}{.59,.44,.09}
\definecolor{mycolor11}{RGB}{231,199,31} 
\definecolor{mycolor12}{RGB}{8,133,161} 
\definecolor{mycolor13}{RGB}{157,188,64} 
\definecolor{mycolor14}{RGB}{194,150,130} 
\definecolor{mycolor15}{RGB}{98,122,157} 
\definecolor{mycolor16}{RGB}{160,160,160} 
\definecolor{mycolor17}{RGB}{115,82,68} 
\definecolor{mycolor18}{RGB}{94,60,108} 
\definecolor{mycolor19}{RGB}{115,82,68} 
\definecolor{mycolor20}{RGB}{255,183,30} 
\definecolor{mycolor21}{RGB}{13, 152, 186}
\definecolor{mycolor22}{RGB}{255, 223, 0}

\definecolor{cambridgeblue}{rgb}{0.64, 0.76, 0.68}
\definecolor{darkturquoise}{rgb}{0.0, 0.81, 0.82}
\definecolor{turquoiseblue}{rgb}{0.0, 1.0, 0.94}
\definecolor{turquoisegreen}{rgb}{0.63, 0.84, 0.71}
\definecolor{turquoise}{rgb}{0.19, 0.84, 0.78}
\definecolor{verdigris}{rgb}{0.26, 0.7, 0.68}
\definecolor{wheat}{rgb}{0.96, 0.87, 0.7}

\theoremstyle{remark} \newtheorem{theorem}{Theorem}
\theoremstyle{remark} 
\theoremstyle{remark} \newtheorem{definition}{Definition}
\theoremstyle{remark} \newtheorem{remark}{Remark}
\theoremstyle{remark} 

\newtheorem{innercustomgeneric}{\customgenericname}
\providecommand{\customgenericname}{}
\newcommand{\newcustomtheorem}[2]{%
  \newenvironment{#1}[1]
  {%
   \renewcommand\customgenericname{#2}%
   \renewcommand\theinnercustomgeneric{##1}%
   \innercustomgeneric
  }
  {\endinnercustomgeneric}
}
\newcustomtheorem{customcorollary}{Corollary}
\newcustomtheorem{customlemma}{Lemma}
\newcustomtheorem{customproposition}{Proposition}

\allowdisplaybreaks

\begin{document}

	\title{Identification for ISI Gaussian Channels}
		\author{\vspace{2mm} \fontsize{12}{12} \selectfont Mohammad Javad Salariseddigh\IEEEauthorrefmark{1} and Christian Deppe\IEEEauthorrefmark{2}
		\vspace{4mm}
		\\
		\fontsize{11}{11} \selectfont \IEEEauthorrefmark{1}Technical University of Darmstadt\\
		\selectfont\IEEEauthorrefmark{2}Technical University of Braunschweig\vspace{0mm}
	}
	
\maketitle

	\begin{abstract}
		We establish lower and upper bounds for the identification capacity of discrete-time Gaussian channels subject to inter-symbol interference (ISI), a canonical model in wireless communication. Our analysis accounts for deterministic encoders under peak power constraint. A principal finding is that, even when the number of ISI taps scales sub-linearly with the codeword length $n$ as $\sim n^{\kappa}$ with $\kappa \in [0,1/2),$ the number of messages that can be reliably identified grows super-exponentially in $n$ as $\sim 2^{(n \log n)R},$ where $R$ denotes the coding rate.
	\end{abstract}
	

	\section{Introduction}
	
	In the identification problem \cite{J85,AD89,AADT20}, appropriate communication protocols enable the receiver to ascertain, with high reliability, whether a specific message, relevant to a designated task, was transmitted. In contrast to Shannon's classical message transmission paradigm \cite{S48,Shannon1957}, where the decoder's aim is to reconstruct the exact message from the entire codebook, the identification framework narrows the decoder's focus to a single, predetermined message. Consequently, the central question is: Was this particular message transmitted or not?
	
	Ahlswede introduced and analyzed the identification problem with randomized encoding for discrete memoryless channels (DMCs) in his seminal work \cite{AD89}. He showed that the DMCs feature a codebook size of \emph{double exponential} in the codeword length, $n$ as $\sim 2^{2^{nR}}$ \cite{AD89}. This scaling highlights a significant divergence from classical channel coding, emphasizing the unique characteristics of the identification framework. Continuous alphabet channels without randomization including Gaussian fading \cite{Salariseddigh_IT,Yuan22}, Poisson with ISI and affine Poisson \cite{Salariseddigh_OJCOMS_23,Salariseddigh25_ITW}, binomial channels \cite{Salariseddigh_Binomial_ISIT} and finite-output channels \cite{Colomer_2025}, exhibit a non-standard codebook size growth, scaling \emph{super-exponentially} in $n$ as $\sim 2^{(n\log n)R}.$ The identification problem has generated substantial interest, particularly within the context of post-Shannon theory, semantic, and goal-oriented communication frameworks; see \cite{Goldreich12,6G_PST,Salariseddigh23_BSC_Future_Internet} for details and possible applications. In particular, it may be adopted to event recognition systems in 6G wireless networks \cite{6G_PST} and molecular communications (MCs) \cite{Salariseddigh-TMBMC}. Discussions on identification code constructions can be found in \cite{Verdu02,Zinoghli24,Vorobyev25,Lengerke25}. One of the widely known models in 6G wireless networks is Gaussian channel with inter-symbol interference (ISI). Such a channel serves as a fundamental model within communication theory for analyzing channels with memory \cite{Proakis01} and is widely used in telephony, satellite, and wireless communication systems \cite{G05}.

	In contrast to memoryless channels \cite{S48,Ahlswede1973MultiwayCC}, channels with memory have received less attention in both the Shannon and Ahlswede settings, primarily due to the absence of single-letter capacity characterizations \cite{Csiszar11}, which complicates both analysis and explicit computation. The single user discrete-time Gaussian channel with ISI (DTGC-ISI) and colored noise is studied by Gallager in \cite{RG68} where he introduced the celebrated water-filling approach for the Shannon capacity. Subsequently, capacity bounds based on independent and identically distributed (i.i.d.) input distributions were established in \cite{Shamai91Info}. The DTGC-ISI capacity under a per-symbol average power constraint was derived in \cite{Hirt88,Hirt02,Tsybakov70}. Capacity regions for the colored-noise broadcast and two-user multiaccess channels were derived in \cite{Goldsmith02} and \cite{Cheng93}. In \cite{Farkas08} the classical method of types for the DTGC-ISI with universal decoding was developed. The DTGC-ISI with stochastic and time-varying ISI taps was studied in \cite{Moshksar24}. Universal, low-complexity, and mismatched decoders were studied in \cite{Huleihel15,Faycal01,Huleihel19,Neeser93}.
	
	The joint identification and channel state estimation problem for a DMC was studied in \cite{labidi2024joint}. A generalized identification framework for slow fading DTGC was considered in \cite{Salariseddigh_22_ITW}. \cite{Salariseddigh26_Colored} studies DTGC-ISI with colored noise. Deterministic identification capacity of the ISI-free DTGC under an average power constraint using multi-layer galaxy codes were studied in \cite{Boche_2025_IEEE,Colomer26}, and the feedback case was shown to support arbitrarily large codebooks in \cite{Wiese22}. The works \cite{Sarkar25,Sarkar25_2} characterized sub-exponential identifiable codebook scaling over noiseless and noisy $q$-ary permutation channels. Beyond its role in modeling bandlimited wireless channels with memory, a variant of the DTGC with ISI, characterized by positive channel impulse response (CIR) taps and non-negative inputs, has also been employed to model molecule-counting receivers in MC systems \cite{Gohari16}. To the best of the authors' knowledge, the identification problem for the DTGC-ISI case has not yet been addressed in the literature. We formulate the identification problem under i.i.d. noise with a deterministic encoder and a peak power constraint, and make the following contributions:
	\begin{itemize}[leftmargin=0mm]
		\item[]	\vspace*{-.2mm} \textbf{\textcolor{blau_2b}{Generalized Gaussian Model:}} In wireless communication systems, often the number of channel taps $K$ can be large, particularly in environments with severe multi-path propagation, e.g., in scenarios where the channel delay spread grows proportionally with the communication window, i.e., $K \propto n.$ To ensure accurate channel estimation and effective equalizer performance, a sufficiently long CIR is required \cite{G05}. Also in MC systems, the CIR may span a large number of taps $K$, particularly in bounded diffusive environments where non-degrading signaling molecules persist and accumulate over time \cite{Gohari16}. Moreover, the effective channel memory $K$ may increase not only with the inherent dispersiveness of the propagation medium but also with the symbol rate. Specifically, higher symbol rates (i.e., smaller sampling intervals $T_s$) lead to a larger number of resolvable channel taps, such that $K \propto T_s^{-1}$. Consequently, it is important to analyze the system's asymptotic behavior in regimes characterized by large symbol rates (large $K$) and long $n,$ to understand performance limits under highly time-dispersive, high-rate operating conditions. To account for this phenomenon, we consider a generalized ISI model covering the ISI-free case ($K=1$), fixed memory channels ($K>1$), and regimes where $K$ grows with $n$.
		\item[] \textbf{\textcolor{blau_2b}{Codebook Scale:}} We establish that the codebook size of the DTGC-ISI for deterministic encoding scales super-exponentially in the codeword length $n$ as $\sim 2^{(n \log n)R},$ where $R$ is the coding rate. This behavior mirrors the scaling behavior as observed for a class of input-continuous alphabet channels including the ISI-free DTGC \cite{Salariseddigh_ITW}, where $K=1.$
		\item[] \textbf{\textcolor{blau_2b}{Memory Scale:}} We formulate the capacity problem for a general DTGC-ISI where the number of ISI taps, $K,$ may scale sub-linearly in the codeword length, $n,$ and show that reliable identification is possible for $K(n,\kappa) = n^{\kappa}$ where $\kappa \in [0,1/2)$ is referred to as the ISI rate.
		\item[] \textbf{\textcolor{blau_2b}{Capacity Bounds:}} We derive lower and upper bounds on $R$ for a class of CIRs $\fh$ expressed as functions of its parameter. In particular, these bounds depend on the length of $\fh,$ namely the scaling of $K$ with $n.$ Then, the lower and upper bounds on the capacity, decrease and increase in $\kappa,$ respectively.
	\end{itemize}
	\vspace{0mm}
	\textbf{Notations:}
	Sans-serif letters $\mathsf{X,Y},\ldots$ stand for alphabet sets. Lowercase letters $x, y, \ldots$ are constants or realizations of random variables (RVs) and uppercase letters $X, Y, \ldots$ represent RVs. Lowercase bold symbols $\fx,\fy,\ldots$ show row vectors. Logarithms are in base $2.$ The sequence of integers from $0$ or $1$ to $M$ is denoted by $[\![M]\!].$ The non-negative real and real numbers are represented by $\mathbb{R}_{+}$ and $\mathbb{R}.$ $\lceil \cdot \rceil$ shows the ceiling. The Gamma function $\Gamma(x)$ for non-negative integer $x$ is $\Gamma (x) = (x-1) ! \triangleq \prod_{l=1}^{x-1} l.$ We adopt the standard \emph{Bachmann-Landau asymptotic notations}. $\norm{\mathbf{x}}$ and  stand for $\ell_2$-norm and $\ell_{\infty}$-norm, respectively. A hypersphere of dimension $n,$ radius $r,$ and center $\mathbf{x}_0$ is defined as $\S_{\fx_0}(n,r) = \{\fx\in\mathbb{R}^n : \norm{\fx-\fx_0} \leq r \}.$ $\mathsf{Q}_{\f0} = \{\fx \in \mathbb{R}^n : |x_t| \leq U, \forall \, t \in [\![n]\!] \}$ indicates an $n$-dimensional cube with edge $U$ and centered at $\mathbf{0} = (0)_{t=1}^n.$ The discrete-time Fourier transform (DTFT) of $(x_t)_{t=0}^{K-1}$ is $X(\omega) = \sum_{k=0}^{K-1} x_ke^{-j\omega k}.$ $Q(.)$ denotes the Gaussian $Q$-function.

	\section{System Model and Coding Preliminaries}
	\label{Sec.SysModel}
	This section presents the DTGC-ISI. We then introduce the adopted system model and establish the coding preliminaries.
	
	\subsection{Discrete-Time Gaussian Channel with ISI}
	We consider an identification-oriented communication model in which the decoder departs from classical message reconstruction \cite{S48} and instead performs a binary hypothesis test: determining, under suitable reliability constraints, whether a prescribed target message was transmitted. A coding scheme is employed to induce a structured encoder-decoder interaction over $n$ channel uses, whereby each message is represented by a length-$n$ codeword transmitted through the channel. We assume that the signal suffers $K$ degrees of ISI memory and experiences an additive i.i.d. Gaussian noise. The memory is characterized by a finite-length sequence $\fh = (h_k)_{k=0}^{K-1}$ referred to as the CIR where $h_k \in \mathbb{R} \,, \forall k$ and $h_0h_{K-1} \neq 0.$ In the sequel, $\bG$ denotes the DTGC-ISI characterized by the CIR $\fh.$ Let $X_t \in \mathbb{R}$ and $Y_t \in \mathbb{R}$ denote RVs modeling the sent and observed symbols at the transmitter and the receiver, respectively. The channel input-output relation is specified letter-by-letter as follows\vspace*{-.5mm}
	\begin{align}
		\label{Eq.Law_Letter}
		Y_t = x_t^{\fh} + Z_t,
	\end{align}
	where $x_t^{\fh} \triangleq \sum_{k=0}^{K-1} h_k x_{t-k}$ is the average signal observed at the receiver. The coefficient $h_k$ is known as the CIR tap at time $k\,, \forall k \in [\![K-1]\!]$ which does not scale with the codeword length $n.$ Moreover $Z_t$ stands for the noise samples with i.i.d. Gaussian distribution with zero mean and finite constant variance $\sigma_Z^2 > 0,$ i.e., $Z_t \overset{\text{\tiny i.i.d.}}{\sim} \N(0,\sigma_Z^2), \forall t \in [\![n]\!].$ Since the channel is dispersive, each output symbol depends on the $K$ most recent inputs, yielding an output vector of length $\bar{n} = n + K - 1$. Therefore, based on the letter-wise channel relation of $\bG$ in \eqref{Eq.Law_Letter}, the channel law admits the following matrix-vector representation.\vspace*{-.5mm}
	\begin{align}
		\mathbf{Y} = \fH \fx + \mathbf{Z},
	\end{align}
	where $\fH$ is a full-rank Toeplitz convolution matrix with $\fH_{\bar{n} \times n} = [h_{i-j}]$ and $h_k=0$ outside $[\![K-1]\!].$ Let $\fZ \allowbreak = \allowbreak (Z_t)_{t=1}^{\bar{n}}$ and $\fY \allowbreak = \allowbreak (Y_t)_{t=1}^{\bar{n}}$ be the noise and output vectors, respectively. Note by setting $\fH \fx \allowbreak = \allowbreak \fx^{\fh} \allowbreak = \allowbreak (x_t^{\fh})_{t=1}^{\bar{n}}$, the noise density reads
	\begin{align}
		f_{\fZ}(\fz) = (2\pi\sigma_Z^2)^{-\bar{n}/2} \exp [ - \| \fy - \fx^{\fh} \|^2 / 2\sigma_Z^2 ].
	\end{align}
	The input strings fulfill $|x_t| \leq P_{\rm max}, \forall t \in [\![n]\!],$ where $P_{\rm max}> 0$ constrains the energy of signal per symbol.
	\subsection{Identification Coding}
	Here, we present the code definition and capacity for $\bG.$
	\begin{definition}[DTGC-ISI Code]
		\label{Def.ISI-Poisson-Code}
		Let $0<e_1,e_2<1$ be fixed error constraints, and let $M(n,R),K(n,\kappa) \in \mathbb N$ denote the number of messages and ISI taps where $n,R$ and $\kappa$ are the codeword length, the coding rate and the ISI rate, respectively. An $(n,\allowbreak M(n,R),\allowbreak K(n,\kappa), \allowbreak e_1, \allowbreak e_2)$-code for $\bG$ consists of a codebook $\mathsf{C} = \{ \fc_i \}_{i=1}^{M_n}$ where $\fc_i = ( c_{i,t} )_{t=1}^n$ fulfills $|c_{i,t}| \leq P_{\rm max} , \forall i \in [\![M]\!],\,\forall t \in [\![n]\!]$ and a collection of decoding regions $\D = \{ \mathsf{D}_i \}_{i=1}^{M_n}.$
		\begin{figure}
			\centering
			\vspace{14mm}
			\scalebox{.88}{
	\tikzstyle{farbverlauf} = [ top color=white, bottom color=white]
	\tikzstyle{block1} = [draw, top color = white, middle color = turquoisegreen!50, rectangle, rounded corners, minimum height=2em, minimum width=2.5em]
	
	\tikzstyle{block2} = [draw, circle,inner sep=0pt, minimum size=5mm, top color = white, middle color = gray!10]
	
	\tikzstyle{block3} = [draw, top color = white, middle color = turquoisegreen!50, rectangle, rounded corners, minimum height=2em, minimum width=2.5em]

	\tikzstyle{input} = [coordinate]
	\tikzstyle{sum} = [draw, circle,inner sep=0pt, minimum size=5mm, top color = white, middle color =gray!10]
	\tikzstyle{arrow}=[draw,->]
	\begin{tikzpicture}[overlay]
		\tikzmath{\x = -5.25;}
		\node[] at (\x,0) (M) {$i$};
		\node[block1,right=.7cm of M] (enc) {Enc};
		\node[block2, right=10mm of enc] (multiplier) {$*$};
		\node[sum, right=25mm of multiplier] (adder) {$+$};
		\node[block3, right=8mm of adder] (dec) {Dec};
		\node[right=.7cm of dec] (Output) {$\text{\small Yes/No}$};
		\node[above=.7cm of Output] (Target) {$j$};
		\node[above=.7cm of adder] (noise) {$Z_t$};
		\node[above=.7cm of multiplier] (CIR) {$\fh = (h_k)_{k=0}^{K-1}$};
		
		\draw[>=open triangle 45,->] (M) -- (enc);
		\draw[>=open triangle 45,->] (enc) --node[above]{$c_{i,t}$} (multiplier);
		\draw[>=open triangle 45,->] (multiplier) --node[above]{$c_{i,t}^{\fh}$} (adder);
		\draw[>=open triangle 45,->] (CIR) -- (multiplier);
		\draw[>=open triangle 45,->] (noise) -- (adder);
		\draw[>=open triangle 45,->] (adder) --node[above]{$Y_t$} (dec);
		\draw[>=open triangle 45,->] (dec) -- (Output);
		\draw[dashed,>=open triangle 45,<-] (dec) |- (Target);
	\end{tikzpicture}
}
			\vspace{3mm}
			\caption{Communication chain for identification in a generic wireless system modeled by $\bG$. The sender maps message $i$ to a codeword $\fc_i = (c_{i,t})_{t=1}^n$. Each channel output symbol results from the convolution of the previous $K$ inputs with the ISI coefficients. Given a query message $j$ and the channel output $\fY = (Y_t)_{t=1}^{\bar{n}}$, the receiver decides whether $j=i$ or not.}
			\vspace{-4mm}
			\label{Fig.SG}
		\end{figure}
		Given a message $i \in [\![M]\!]$, the encoder sends $\mathbf{c}_i$, and the decoder asks was a target message $j,$ sent or not? See Figure \ref{Fig.SG}. In this framework, two decoding error events may occur: the false rejection of the sent (true) codeword and the false acceptance of an incorrect codeword. These events correspond to type I and type II errors, respectively, and are formally characterized as follows:
		\begin{align*}
			\hspace{-1mm} P_{e,1}(i) \hspace{-.6mm} = \hspace{-.6mm} 1 \hspace{-.6mm} - \hspace{-1mm} \int_{\mathsf{D}_i} \hspace{-2mm} f_{\fZ}(\fy - \fc_i^{\fh}) \, d\fy,\, P_{e,2}(i,j) \hspace{-1mm} = \hspace{-1mm} \int_{\mathsf{D}_j} \hspace{-2mm} f_{\fZ}(\fy - \fc_i^{\fh}) \, d\fy .
		\end{align*}
		The maximal error probabilities read $\lambda_1 \triangleq  \max_{i}P_{e,1}(i)$ and $\lambda_2 \triangleq \max_{ i\neq j} P_{e,2}(i,j),$ and satisfy $\lambda_1 \leq e_1$ and $\lambda_2 \leq e_2.$
		\qed
	\end{definition}
	\begin{definition}[DTGC-ISI Capacity]A rate $R>0$ is called achievable if for every $e_1, \allowbreak e_2>0$ and sufficiently large $n,$ there exists an $(n,\allowbreak M(n\allowbreak,R),\allowbreak K(n,\allowbreak \kappa), \allowbreak e_1, \allowbreak e_2)$-code. The operational identification capacity of $\bG$ is defined as the supremum of all achievable rates, and is denoted by $\mathbb{C}_{\text{I}}(\bG).$
	\qed
	\end{definition}
	\begin{remark}
		For simplicity, throughout the paper we use $M$ and $K$ in place of $M(n,R)$ and $K(n,\kappa),$ respectively, unless stated otherwise. We explicitly indicate when $K \geq 1$ is a fixed constant by specifying $\kappa=0$ or $\kappa \to 0$ as $n \to \infty.$
	\end{remark}
	\section{Identification Capacity of The DTGC-ISI}
	\label{Sec.Res}
	We begin by posing two fundamental questions that frame the design and analysis of wireless systems with ISI, thereby establishing the framework for the developments that follow. The number of ISI taps is defined as $K = \left\lceil T_{\rm c}/T_s \right\rceil,$	where $T_{\rm c}$ and $T_s$ denote the channel delay spread duration and the symbol duration, respectively \cite{G05}. In practice, $T_{\rm c},$ and thus $K,$ may be large. The results herein hold for both fixed $K$ and regimes in which $K$ grows with $n$ (e.g., as $T_s$ decreases). Next, we present the capacity theorem and sketch its proof, with particular emphasis on addressing the following questions:
	\begin{itemize}[leftmargin=0mm]
		\item[] \textbf{\textcolor{blau_2b}{Size of Identifiable Messages:}} What is the channel's asymptotic message distinguishability, i.e., maximal number of the codewords that can be reliably decoded as $n \to \infty$?
		\item[] \textbf{\textcolor{blau_2b}{Characteristics of CIR:}} What structure must the CIR taps $h_k$ satisfy for reliable identification? Specifically, what properties may a valid CIR sequence exhibit (e.g., sparsity, convergence, or spectral characteristics)? Must the CIR depth $K$ remain fixed, or can reliable identification still hold for growing depth, i.e., $K = O(n^{\kappa})$ with $\kappa \in [0,1]$?
	\end{itemize}
	\subsection{Main Results}
	\label{Subsec.Main_Results}
	The identification capacity depends on the properties of the CIR $\fh$ and how its features scale with $n$. We impose the following conditions, where \textbf{C1} is required for both achievability and converse, and \textbf{C2} is needed only for the achievability.
	\begin{itemize}[leftmargin=*]
		\item \textbf{\textcolor{blau_2b}{C1 (Stability):}} We assume that the ISI channel $\bG$ is stable, i.e., the CIR meets an absolute summable condition: $\sum_{k=0}^{K-1} |h_k| < \infty,$ which implies that the individual ISI taps $h_k$ fulfill: $|h_k| \leq L < \infty\,, \forall k \in [\![K-1]\!]$ where $L$ is a finite fixed constant independent of $h_k$ and $n.$
		\item \textbf{\textcolor{blau_2b}{C2 (Frequency Spectrum):}} Let $H(\omega)$ be the frequency spectrum, i.e., the DTFT transform of the CIR vector $\fh.$ Then, we assume that the kernel $H(\omega)$ is bounded away from zero, that is, $\inf_{\omega\in[-\pi,\pi]}|H(\omega)| \triangleq \sqrt{2\pi} H_{\rm min} > 0$. See Remark \ref{Rem.Spec_Nulls} for cases where $\inf_{\omega} |H(\omega)| = 0.$
	\end{itemize}
	Next, we formally present the capacity theorem for $\bG.$
	\begin{theorem}
		\label{Th.ISI-Gauss-Cap}
		Consider the DTGC-ISI, $\bG,$ with CIR $\fh$ fulfilling conditions \textbf{C1}-\textbf{C2} and assume that the number of ISI channel taps grows sub-linearly with the codeword length, i.e., $K(n,\kappa) = n^{\kappa},$ where $\kappa \in [0,1/2).$ Then, the identification capacity of $\bG$ subject to peak power constraint according to Definition~\ref{Def.ISI-Poisson-Code} and in the super-exponential codebook size scale, i.e., $M(n,R) = 2^{(n\log n)R},$ reads
		\begin{align}
			\label{Ineq.LU}
			\frac{1 - 2\kappa}{4} \leq \mathbb{C}_{\rm I}(\bG) \leq 1+\kappa.
		\end{align}
	\end{theorem}
	\begin{proof}
		Proofs for the achievability and the converse are provided in Subsections~\ref{Subsec.Achievability} and \ref{Subsec.Converse}, respectively.
	\end{proof}
	
	\begin{remark}
		The result in Theorem~\ref{Th.ISI-Gauss-Cap} comprises the following three special cases in terms of $K$:
		\begin{itemize}[leftmargin=*]
			\item \textbf{\textcolor{blau_2b}{$K=1$}}: This case corresponds to an ISI-free setup, which holds when the symbol duration is sufficiently large ($T_s \geq T_{\rm c}$). Under this condition $K=1$ and $\kappa=0$ and the bounds in \eqref{Ineq.LU} recover the previously known results in \cite[Th.~10]{Salariseddigh_IT} by setting $\kappa = 0$ and fading coefficient $G =1.$
			\item \textbf{\textcolor{blau_2b}{Constant $K>1$}}: When $T_s$ is fixed and $T_s < T_{\rm c}$, the constant $K>1$ implies that $\kappa \to 0$ as $n \to \infty.$ Specifically, $\kappa$ decreases monotonically to zero with rate $\kappa = \Theta(1/\ln n).$ Interestingly, Theorem~\ref{Th.ISI-Gauss-Cap} shows that the DTGC-ISI capacity bounds coincide with those of the ISI-free DTGC in \cite{Salariseddigh_ITW}.
			\item \textbf{\textcolor{blau_2b}{Growing $K$}}: Theorem~\ref{Th.ISI-Gauss-Cap} shows that reliable identification remains feasible even when $K$ is strictly monotonic in $n$ and $\kappa,$ for $\kappa \in (0,1/2),$ and $K \sim n^{\kappa}.$ The effect of $\kappa$ appears in the bounds of \eqref{Ineq.LU}, where larger $\kappa$ tightens the lower bound downward and loosens the upper bound upward.
		\end{itemize}
	\end{remark}
	\subsection{Achievability}
	\label{Subsec.Achievability}
	The achievability proof consists of the following steps.
	\begin{itemize}[leftmargin=*]
		\item Coding scheme based on the optimal arrangement of hyperspheres within the eligible input space.
		\item Utilizing a maximum likelihood distance decoder to characterize the \emph{achievable} rates induced by the coding scheme.
	\end{itemize}
	
	\textbf{\textcolor{blau_2b}{Codebook Construction:}}
	In the following, we construct the codebook $\mathsf{C} = \{ \fc_i \} \subset \mathbb{R}^n,$ with $i \in [\![M]\!],$ referred to as the original codebook where we restrict ourselves to the codewords that meet the condition $|c_{i,t}| \leq P_{\rm max}.$ This codebook induces a secondary, \emph{convoluted codebook}, whose properties depend on the structure of the ISI channel and the original codebook $\mathsf{C}$. We denote this convoluted codebook by	$\mathsf{C}^{\fh} = \{ \mathbf{c}_i^{\fh} \} \subset \mathbb{R}^{\bar{n}} ,$ with $i \in [\![M]\!],$ where each $\mathbf{c}_i^{\fh} \triangleq (c_{i,t}^{\fh})_{t=1}^{\bar{n}}$ is referred to as a convoluted codeword whose symbols are formed as a linear combination (convolution) of the $K$ most recent symbols of codeword $\fc_i \triangleq (c_{i,t})_{t=1}^n$ and CIR vector $\fh,$ i.e., $c_{i,t}^{\fh} \triangleq \sum_{k=0}^{K-1} h_k c_{i,t-k}\,, \forall t \in [\![\bar{n}]\!]$ where $c_{i,t} = 0,\, \forall t \leq 0.$ With this notation, we define the feasible original codebook and the corresponding convoluted codebook as follows:
	\begin{itemize}[leftmargin=*]
		\item $\mathsf{C} \hspace{-.3mm} \triangleq \hspace{-.3mm} \big\{ \fc_i \in \mathbb{R}^n:\; |c_{i,t}| \leq P_{\rm max} , \forall i \in [\![M]\!], \forall t \in [\![n]\!] \big\}.$
		\item $\mathsf{C}^{\fh} \hspace{-.3mm} \triangleq \hspace{-.3mm} \big\{ \fc_i^{\fh} \in \mathbb{R}^{\bar{n}} :\, \mathbf{c}_i^{\fh} = (c_{i,t}^{\fh} )_{t=1}^{\bar{n}}, \, \fc_i \in \mathsf{C}, \forall i \in [\![M]\!] \big\}.$
	\end{itemize}
	
	\textbf{\textcolor{blau_2b}{Rate Analysis:}} We use a packing arrangement of non-overlapping hyperspheres of radius $r_0 = \sqrt{\bar{n}\epsilon_n}$ in the hypercube $\mathsf{C}$ with edge $P_{\rm max},$ where 
	\begin{equation}
		\epsilon_n = a / (H_{\rm min}^2 n^{( 1 - ( 2\kappa + b)) / 2}),
	\end{equation}
	 with $a,b > 0$ being fixed and arbitrarily small constants, respectively. Let $\mathscr{S}$ denote a sphere packing, i.e., an arrangement of $M$ non-overlapping spheres $\S_{\fc_i}(n,r_0),\, i \in [\![M]\!],$ that are packed in the hypercube $\mathsf{Q}_{\f0}(n,P_{\rm max}).$ In contrast to conventional sphere packing in coding theory \cite{CHSN13}, we require that the centers of disjoint spheres are in $\mathsf{Q}_{\f0}(n,P_{\rm max}),$ and that each sphere exhibits a non-empty intersection with $\mathsf{Q}_{\f0}(n,P_{\rm max}).$ Note that, the packing density \cite{CHSN13} is
	\begin{align}
		\label{Eq.Density}
		\Updelta_n(\mathscr{S}) \triangleq \frac{\text{Vol}\big(\mathsf{Q}_{\f0}(n,P_{\rm max}) \cap \bigcup_{i=1}^{M}\S_{\fc_i}(n,r_0) \big)}{\text{Vol}\big(\mathsf{Q}_{\f0}(n,P_{\rm max}) \big)} .
	\end{align}
	A saturated packing argument is invoked here, analogous to the approach used in the Minkowski–Hlawka theorem for saturated packings, as referenced in \cite{CHSN13}. Specifically, consider a saturated packing of spheres \(\bigcup_{i=1}^{M(n,R)} \S_{\fc_i}(n,r_0)\) with radius \( r_0 = \sqrt{\bar{n}\epsilon_n} \), embedded within the hypercube \(\mathsf{Q}_{\f0}(n,P_{\rm max})\). In general, the volume of a hypersphere of radius \( r \) is given by \cite[Eq.~(16)]{CHSN13}, 
	\begin{align}
		\label{Eq.VolS}
		\text{Vol}\left(\S_{\fc_i}(n,r)\right) = \frac{(\sqrt{\pi}r)^{n} }{\Gamma(n/2+1)} .
	\end{align}
	Note that density of such arrangement fulfills \cite[Sec.~IV]{Salariseddigh-TMBMC}
	\begin{align}
		\label{Ineq.Density}
		2^{-n} \leq \Updelta_n(\mathscr{S}) \leq 2^{-0.599n} .
	\end{align}
	We assign a codeword to the center $\fc_i$ of each small hypersphere. The codewords meet $\| \fc_i \|_{\infty} \leq P_{\rm max}.$ Since the volume of each sphere is equal to $\text{Vol}(\S_{\fc_1}(n,r_0))$ and the centers of all spheres lie inside the cube $\mathsf{Q}_{\f0}(n,P_{\rm max}),$ the total number of packed spheres, $M,$ is bounded by
	\begin{align}
		M & = \frac{\text{Vol}\big(\bigcup_{i=1}^{M}\S_{\fc_i}(n,r_0)\big)}{\text{Vol}(\S_{\fc_1}(n,r_0))} 
		\nonumber\\&
		\geq \frac{\text{Vol}\big(\mathsf{Q}_{\f0}(n,P_{\rm max}) \cap \bigcup_{i=1}^{M}\S_{\fc_i}(n,r_0) \big)}{\text{Vol}(\S_{\fc_1}(n,r_0))}
		\nonumber\\&
		\stackrel{(a)}{\geq} 2^{-n} \cdot \frac{(2P_{\rm max})^n}{\text{Vol}(\S_{\fc_1}(n,r_0))} ,
	\end{align}
	where $(a)$ exploits \eqref{Eq.Density} and \eqref{Ineq.Density}. The above bound can be simplified as follows
	\begin{align}
		\label{Eq.Log_M_0}
		\log M
		& \geq \log \big( P_{\rm max}^n \big) - \log \big(\text{Vol}(\S_{\fc_1}(n,r_0)) \big) - n/2
		\nonumber\\&
		\stackrel{(a)}{\geq} n \log (2P_{\rm max}) - n \log r_0 + \floor{n/2} \log \floor{n/2}
		\nonumber\\&
		- \floor{n/2} \log e + o \big( \floor{n/2} \big) - n/2 ,
	\end{align}
	where $(a)$ uses \eqref{Eq.VolS} and Stirling's approximation, namely, $\log n! = n \log n - n \log e + o(n)$ \cite[P.~52]{F66} with setting $n$ with $\floor{n/2} \in \mathbb{Z},$ and since
	\begin{align}
		\label{Ineq.Gamma_LB}
		\Gamma (( n/2) + 1 ) & \stackrel{(a)}{=} (n/2) \cdot \Gamma \left( n/2 \right)
		\nonumber\\&
		\stackrel{(b)}{\geq} \floor{n/2} \cdot \Gamma \big( \floor{n/2} \big) \stackrel{(c)}{\triangleq} \floor{n/2} ! ,
	\end{align}
	where $(a)$ holds since $\Gamma(x+1)=x\Gamma(x),\, \forall x \in \mathbb{R}_{+},$ $(b)$ follows from $\floor{n/2} \leq n/2$ and monotonicity of the Gamma function for $n \geq 4,$ $(c)$ holds uses $\Gamma \big( \floor{n/2} \big) = (\floor{n/2} - 1 ) !$ for $n \geq 4.$ Now, observe
	\begin{align}
		r_0 = \allowbreak \sqrt{\bar{n}\epsilon_n} & = \sqrt{n\epsilon_n(1+(K-1)/n)}
		\nonumber\\&
		\sim \sqrt{n\epsilon_n} 
		= \sqrt{a} H_{\rm min}^{-1} n^{(1 + 2\kappa + b) / 4}.
	\end{align}
	Thereby, we obtain the following bound on the logarithm of the number of packed spheres,
	\begin{align}
		& \log M \stackrel{(a)}{\geq} n \log \Big( \frac{2P_{\rm max}H_{\rm min}}{\sqrt{ae}} \Big) - \frac{n}{2} \log (n\epsilon_n(1+K/n))
		\nonumber\\&
		+ ( (n/2) - 1 ) \log ( (n/2) - 1 ) + \mathcal{O}(n) =
		\nonumber\\&
	    \left( \frac{2 - (1 + 2\kappa + b)}{4} \right) n \log n + n \log \Big( \frac{2P_{\rm max} H_{\rm min}}{\sqrt{ae}} \Big)  + \mathcal{O}(n) ,
		\label{Eq.Log_M}
	\end{align}
	where $(a)$ follows by $\floor{n/2} > (n/2) - 1$ and  $\log(t-1) \geq \log t - 1$ for $t \geq 2$ and $\floor{n/2} \leq (n/2)$ for integer $n.$ Consequently, the leading-order term in \eqref{Eq.Log_M} scales as $n \log n$. To ensure the derived lower bound on the achievable rate $R$ remains finite in the asymptotic limit, this expression dictates the requisite scaling for $M$. Specifically, $M$ must scale according to $M = 2^{(n \log n)R}$. Thereby,
	\begin{equation*}
		\resizebox{.491\textwidth}{!}{
			$\begin{aligned}
				\hspace{-1mm} R \geq \hspace{-.6mm} \frac{1}{n \log n} \hspace{-.7mm} \left[ \left( \frac{2 - (1 + 2\kappa + b)}{4} \right) \hspace{-.4mm} n \log n + n \log \left( \frac{2P_{\text{max}}H_{\rm min}}{\sqrt{ae}} \right) + \mathcal{O}(n) \right]
			\end{aligned}$}
	\end{equation*}
	which tends to $(1-2\kappa)/4$ when $n \to \infty$ and $b \rightarrow 0.$

	\begin{customlemma}{1}[Minimum Distance of the Convoluted Codebook]
		\label{Lem.Min_dis_con_cb}
		 Minimum distance of $\mathsf{C}^{\fh}$ fulfills: $	\| \fc_i^{\fh} - \fc_j^{\fh} \| \geq H_{\rm min} \| \fc_i - \fc_j \|$ where $\inf_{\omega \in [-\pi,\pi]} |H(\omega)| / \sqrt{2\pi} \triangleq H_{\rm min} > 0.$
		\begin{proof}
			Observe that employing the Parseval's theorem \cite{Oppenheim99} for a causal and finite-length sequence $\fc_i^{\fh} - \fc_j^{\fh},$ we obtain:
				\vspace{-2mm}
			\begin{align*}
				\hspace{-2mm} \big\| \fc_i^{\fh} - \fc_j^{\fh} \big\|^2
				\hspace{-1mm} \stackrel{(a)}{\geq} \hspace{-1mm} H_{\rm min}^2 \hspace{-1.3mm} \int_{-\pi}^{\pi} \hspace{-3.5mm} |C_i(\omega) - C_j(\omega) |^2 d \omega
				\hspace{-1mm} \stackrel{(b)}{=} \hspace{-1mm} H_{\rm min}^2 \| \fc_i - \fc_j \|^2, \hspace{-2mm}
			\end{align*}
			where $(a)$ uses the Parseval's theorem, linearity of the DTFT and since the DTFT of convolution is product of DTFTs and definition of $H_{\rm min}$ and $(b)$ uses the Parseval's theorem.
		\end{proof}
	\end{customlemma}
	Next, we present two lemmas giving \emph{sufficient} conditions under which $H_{\rm min}$ in Lemma~\ref{Lem.Min_dis_con_cb} remains strictly positive. We first consider constant $K$, covered by Lemma~\ref{Lem.Constant_K}, and then the general regime where $K$ grows with $n$, parameterized by $\kappa\in(0,1/2)$. Unlike the fixed-$K$ case, the spectral minimum may asymptotically vanish; Lemma~\ref{Lem.Increasing_K} provides a sufficient condition to preclude this behavior.
	\begin{customlemma}{2}[Constant $K$-Regime]
	\label{Lem.Constant_K}
		Consider $\bG$ with constant number of ISI taps, i.e., $K(n,\kappa) = \mathcal{O}(1)$ ($\kappa=0$ or $\kappa \to 0 \allowbreak \text{ as } n \to \infty.$) Suppose $H(\omega)$ has no zero on the unit circle, and let $z_k$ be the roots of $z^{K-2} H(z) = h_0 \allowbreak \prod_{k=0}^{K-2} \allowbreak (z - z_k).$ Then, $H_{\rm min}$ is uniformly bounded away by zero, i.e., there exists a constant $H_{\rm LB} > 0$ such that $H_{\rm min} \geq H_{\rm LB} / \sqrt{2\pi} > 0.$
	\end{customlemma}
	\begin{proof}
		The proof is provided in Appendix \ref{App.LB_FS}.
	\end{proof}

	\begin{customlemma}{3}[Increasing $K$-Regime]
		\label{Lem.Increasing_K}
		Consider $\bG$ with increasing $K,$ where $\kappa \in (0,1/2)$ is fixed. Suppose that the CIR $\fh$ satisfies $|h_0| > \sum_{k=1}^{K-1} |h_k|$ and $m \triangleq |h_0| - \sum_{k=1}^{K-1} |h_k| > 0$ where $m$ is a fixed constant. Then, 
		\begin{equation}
			\| \fc_i^{\fh} - \fc_j^{\fh} \| \geq m \| \fc_i - \fc_j \|.
		\end{equation}

	\end{customlemma}
	\begin{proof}
		The proof is provided in Appendix \ref{App.Increasing_K}.
	\end{proof}
	\begin{remark}
		Lemmas~\ref{Lem.Constant_K} and \ref{Lem.Increasing_K} address distinct asymptotic parameter regimes: fixed $K$ and growing $K$, respectively. Both provide sufficient conditions ensuring $H_{\rm min}$ remains bounded away from zero, establishing a non-empty admissible parameter space in each regime. These conditions are constructive rather than necessary for achievability.
	\end{remark}
	
	\begin{remark}
		Scenarios in the communication systems for which the CIR is characterized by a high-power first tap followed by much weaker subsequent taps (multipath), may be adopted to Lemma \ref{Lem.Increasing_K}, i.e., where $\bG$ features a convolution matrix $\fH$ which enjoys a strictly positive smallest singular value. For example, in satellite communication or indoor environments when a single direct signal path is significantly stronger than sum of other multipath reflections \cite{G05}. In such scenarios Rician fading models wireless channels with a direct line-of-sight path combined with scattered multipath components.
	\end{remark}
	\begin{remark}
		\label{Rem.Spec_Nulls}
		Observe that, in relation to Lemma~\ref{Lem.Min_dis_con_cb}, deriving a non-trivial lower bound on $\| \mathbf{c}_i^{\mathbf{h}} - \mathbf{c}_j^{\mathbf{h}}\|,$ reveals a fundamental trade-off between assumptions on the CIR and structural constraints on the codebook. By Parseval's theorem, the squared Euclidean distance admits a frequency-domain representation in which it depends on the product of the channel spectrum and the spectrum of the codeword difference. A direct approach to lower bounding this quantity requires that the channel frequency response be uniformly bounded away from zero, that is, its infimum is strictly positive as accomplished in Lemma~\ref{Lem.Min_dis_con_cb}. If this condition is relaxed, allowing the channel spectrum to exhibit nulls or vanishing behavior over subsets of the frequency domain, i.e.,
		\begin{align}
			\label{Eq.Null_Spec}
			\inf_{\omega} |H(\omega)| = 0
		\end{align}
		as is typical for frequency-selective channels in practical scenarios~\cite{G05}, then additional structure must be imposed on the codebook. In particular, constraints on the Toeplitz matrix formed by stacking the codewords as rows, are required to ensure that their spectral support avoids the null space of the channel. Such restrictions convolutedly reduce the admissible signal space for sphere-packing arguments, thereby shrinking the feasible packing region. Consequently, this may lead to a degradation in achievable rates or a reduction in the scaling of the attainable codebook size. A systematic characterization of this trade-off, and the design of constructions that optimally balance channel uncertainty and codebook structure, remain open problems and constitute a promising direction for future research.
		
		Observe that if \eqref{Eq.Null_Spec} holds, then the channel severely suppresses the received signal over certain frequency bands. In this regime, there exist frequencies at which the channel gain is arbitrarily small, causing corresponding spectral components of the transmitted signal to be strongly attenuated or convolutedly eliminated. As a result, the channel convolution operator becomes ill-conditioned and fails to be bounded below. That is, Lemma \ref{Lem.Increasing_K} does not work since
		\begin{equation}
			\sigma_{\rm min}(\fH)=0.
		\end{equation}
		Consequently, distinct transmitted signals can be mapped to nearly indistinguishable received signals after propagation through the channel. This leads to a collapse of separation in the signal space and ultimately results in a loss of identifiability at the receiver for distinct messages. \qed
	\end{remark}

	\textbf{\textcolor{blau_2b}{Encoding:}} For message $i \in [\![M]\!],$ sender transmits $\fx = \fc_i.$
	
	\textbf{\textcolor{blau_2b}{Decoding:}}
	Let $e_1, e_2, \eta_0, \zeta_0, \zeta_1 > 0$ be arbitrarily small constants. For concise analysis, we adopt the following conventions:
	\begin{itemize}[leftmargin=*]
		\item $Y_t(i) = c_{i,t}^{\fh} + Z_t,\, \forall t \in [\![\bar{n}]\!]$ is channel output given $\fx=\fc_i.$
		\item The output vector is $\fY(i) = (Y_t(i))_{t=1}^{\bar{n}}$ with $\bar{n} = n+K-1.$
		\item $c_{i,t}^{\fh} \triangleq \sum_{k=0}^{K-1} h_k c_{i,t-k}, \, \forall t \in [\![\bar{n}]\!]$ is the convoluted symbol.
		\item $\delta_n = 4a / 3n^{(1-(2\kappa + b))/2}$ is the \emph{decoding threshold} where $a,b>0$ are fixed and arbitrary constants, respectively.
		\item The spectrum is non-vanishing, i.e., $\inf_{\omega \in [-\pi,\pi]} \hspace{-.5mm}  |H(\omega)| \hspace{-.7mm} > \hspace{-.7mm} 0.$
	\end{itemize}
	To identify if message $j$ was sent, the decoder checks if $\mathbf{y}$ belongs to decoding set $\mathsf{D}_j = \{ \fy \in \mathbb{R}^{\bar{n}} \,:\; T(\fy,\fc_j^{\fh}) \leq \delta_n \}$ where $T(\fy,\fc_j^{\fh}) = \bar{n}^{-1} \| \fy -  \fc_j^{\fh} \|^2 - \sigma_Z^2$ is referred to as the \emph{decoding measure}. Next, we adopt the following definitions:
	\begin{itemize}[leftmargin=*]
		\item $T(\fY(i),\fc_j^{\fh}) \hspace{-.6mm} = \hspace{-.6mm} \beta - \alpha$ with $\beta \hspace{-.6mm} \triangleq \hspace{-.6mm} \bar{n}^{-1} \| \fY(i) -  \fc_j^{\fh} \|^2$ and $\alpha \hspace{-.6mm} \triangleq \hspace{-.6mm} \sigma_Z^2.$
		\item $\beta_1 \triangleq \bar{n}^{-1} \big( \| \fZ \|^2 + \| \fc_i^{\fh} - \fc_j^{\fh} \|^2 \big).$
		\item $\beta_2 \triangleq 2\bar{n}^{-1} \sum_{t=1}^{\bar{n}} \big(c_{i,t}^{\fh} - c_{j,t}^{\fh} \big)Z_t.$
		\item $\E_0 = \big\{ \fZ \in \mathbb{R}^{\bar{n}} \;:\, \big|  2\bar{n}^{-1} \sum_{t=1}^{\bar{n}} \big( c_{i,t}^{\fh} - c_{j,t}^{\fh} \big) Z_t \big| > \delta_n \big\}.$
		\item $\E_1 = \big\{ \fZ \in \mathbb{R}^{\bar{n}} \;:\, \bar{n}^{-1} \big( \| \fZ \|^2 + \| \fc_i^{\fh} - \fc_j^{\fh} \|^2 \big) - \sigma_Z^2 \leq 2\delta_n \big) \big\}.$
		\item $\E_2 = \big\{ \fZ \in \mathbb{R}^{\bar{n}} \;:\, \bar{n}^{-1} \| \fZ + \fc_i^{\fh} - \fc_j^{\fh} \|^2 - \sigma_Z^2 \leq \delta_n \big\}.$
	\end{itemize}

	\textbf{\textcolor{blau_2b}{Type I:}}
	The type I errors occur when the transmitter sends $\fc_i,$ yet $\fY \notin \mathsf{D}_i.$ $\forall i \in [\![M]\!],$ the type I error probability reads 
		\vspace{-1mm}
	\begin{align}
		\label{Eq.TypeIError}
		P_{e,1}(i) = \Pr\big( \fY(i) \in \mathsf{D}_i^c \big) = \Pr\big( T(\fY(i),\fc_i^{\fh}) > \delta_n \big).
	\end{align}
	To bound $P_{e,1}(i),$ we perform the Chebyshev's inequality, i.e.,
	\begin{equation}
		\Pr( T(\fY(i),\fc_i^{\fh}) > \delta_n ) \leq \text{Var} [ T(\fY(i),\fc_i^{\fh}) ]/\delta_n^2.
	\end{equation}
	Linearity of the expectation gives $\mathbb{E} [T(\fY(i),\fc_i^{\fh})] = 0.$ Next, invoking statistical independence of noise and $\text{Var}[Z_t^2] \leq \mathbb{E}[Z_t^4]\,, \forall t \in [\![\bar{n}]\!],$ gives $\text{Var}[ T(\fY(i),\fc_i^{\fh})] \leq 3\sigma_Z^4/\bar{n}.$ Thereby,
	\vspace{-1mm}
	\begin{align}
		\label{Ineq.TypeI_Final}
		\hspace{-2mm} P_{e,1}(i) \leq \frac{\text{Var} \big[ T(\fY(i),\fc_i^{\fh}) \big]}{\delta_n^2} \leq \frac{3\sigma_Z^4}{\bar{n}\delta_n^2} \stackrel{(a)}{\leq} \frac{1.6875\sigma_Z^4}{a^2\bar{n}^{2\kappa + b}} \triangleq \eta_0, 
	\end{align}
	where $(a)$ uses $\delta_n = 4a / 3n^{(1-(2\kappa + b))/2}.$ Hence, $P_{e,1}(i) \leq \eta_0 \allowbreak \leq \allowbreak e_1$ holds for sufficiently large $n$ and arbitrarily small $e_1 > 0.$
	
	\textbf{\textcolor{blau_2b}{Type II:}}
	We examine type II errors, i.e., when $\fY \in \mathsf{D}_j$ while the transmitter sent $\fc_i$ with $i \neq j \,.$ Then, for every $i,j \in [\![M]\!],$ the type II error probability is given by
	\begin{align}
		P_{e,2}(i,j) = \Pr \big( T(\fY(i);\fc_j^{\fh}) \leq \delta_n \big) = \Pr(\E_2) .
		\label{Eq.Pe2G}
	\end{align}
	Since $2\bar{n}^{-1} \sum_{t=1}^{\bar{n}} \big( c_{i,t}^{\fh} - c_{j,t}^{\fh} \big) Z_t
	\sim \N(0,4\sigma_Z^2 \|\fc_i^{\fh} - \fc_j^{\fh} \|^2 / \bar n^2  )$ therefore
	\begin{equation}
		\Pr(\E_0) = 2Q( \bar{n} \delta_n / 2\sigma_Z \|\fc_i^{\fh}-\fc_j^{\fh}\| ).
	\end{equation}
	Using the standard tail bound for the Gaussian $Q$-function, $2Q(x)\leq e^{-x^2/2}$, and $\|\fc_i^{\fh}-\fc_j^{\fh}\|^2 \le 4\bar nK^2L^2P_{\rm max}^2,$ we obtain an exponential upper bound as follows
	\begin{align}
		\label{Ineq.Event_E0_2}
		\Pr(\E_0) \allowbreak \leq \allowbreak \exp ( - a^2n^{2\kappa+b} / 18\sigma_Z^2 n^{2\kappa}L^2P_{\rm max}^2) \triangleq \zeta_0 .
\end{align}
Next, given the complementary event $\E_0^c,$ we obtain
	\begin{equation}
		2\sum_{t=1}^{\bar{n}} \big( c_{i,t}^{\fh} - c_{j,t}^{\fh} \big) Z_t > - \bar{n}\delta_n.
	\end{equation}
	Now, applying law of total probability to $\E_2$ for $(\E_0,\E_0^c)$ gives
	\vspace{-5mm}
	\begin{align}
		\hspace{-3mm} P_{e,2} (i,j) \stackrel{(a)}{\leq} \Pr(\E_0) + \Pr( \E_2 \cap\,{\E_0^c})
		\stackrel{(b)}{\leq} \Pr( \E_0) + \Pr( \E_1), \hspace{-2mm}
		\label{Eq.TypeIIError-E_0+E_1} 
	\end{align}
	where $(a)$ uses $\E_2 \cap \E_0 \subset \E_0$ and $(b)$ holds by the following
	\begin{equation*}
		\Pr(\{\beta_1 - \alpha \leq \delta_n-\beta_2\}
	\cap\{| \beta_2 | \leq \delta_n\}) \leq \Pr(\{\beta_1 - \alpha \leq 2\delta_n\}),
	\end{equation*}
	where the last inequality uses $\delta_n-\beta_2\leq 2\delta_n$ given $| \beta_2 | \leq \delta_n.$ By construction, codewords are separated by at least
	\begin{equation}
		2r_0=2\sqrt{\bar{n}\epsilon_n} ,
	\end{equation}
	yielding the following bound on $\Pr(\E_1)$. Therefore, exploiting Lemma \ref{Lem.Min_dis_con_cb}, we obtain $\| \fc_i^{\fh} - \fc_j^{\fh} \|^2 \geq 4H_{\rm min}^2\bar{n}\epsilon_n.$ Thus, we establish the following bound for $\E_1:$
		\vspace{-2mm}
	\begin{align}
		\label{Ineq.Event_E1}
		\Pr(\E_1) \leq \frac{\sum_{t=1}^{\bar{n}} \text{Var}[Z_t^2]}{\bar{n}^2\delta_n^2}
		\stackrel{(a)}{\leq} \frac{3\sigma_Z^4}{\bar{n}\delta_n^2} \stackrel{(b)}{\leq} \frac{27\sigma_Z^4}{16a^2n^{b+2\kappa}}	\triangleq \zeta_1,
	\end{align}
	where $(a)$ uses the Chebyshev's inequality and $(b)$ follows by similar arguments as provided in \eqref{Ineq.TypeI_Final}. Therefore, employing the upper bounds given in \eqref{Ineq.Event_E0_2}, \eqref{Eq.TypeIIError-E_0+E_1} and \eqref{Ineq.Event_E1},  yields
	\begin{align}
		P_{e,2}(i,j) \leq \Pr(\E_0) + \Pr(\E_1) \leq \zeta_0 + \zeta_1 \leq e_2,
	\end{align}
	hence, $P_{e,2}(i,j)\leq e_2$ for sufficiently large $n.$ Thus, $\forall e_1,e_2 > 0$ and sufficiently large $n$, an $(n,M(n,R),K(n,\kappa),e_1,e_2)$-code exists, completing the achievability proof of Theorem~\ref{Th.ISI-Gauss-Cap}.

	\subsection{Upper Bound (Converse Proof)}
	\label{Subsec.Converse}
	The converse proof is structured as follows. Lemma~\ref{Lem.Converse} bounds the minimum distance of any convoluted codebook, yielding an upper bound on the size of the identification codebook. For conciseness, we introduce following conventions:
	\begin{itemize}[leftmargin=*]
		\item \vspace*{-.3mm}$A_{\rm max} = KLP_{\rm max} = \mathcal{O}(n^{\kappa}).$
		\item $\mathsf{C}_{\text{\tiny conv}} \triangleq \big\{ \fc_i \in \mathbb{R}^n:\; |c_{i,t}| \leq P_{\rm max} , \forall \, i \in [\![M]\!],\, \forall t \in [\![n]\!] \big\} .$
		\item $\mathsf{C}_{\text{\tiny conv}}^{\fh} \hspace{-.3mm} \triangleq \hspace{-.3mm} \big\{ \fc_i^{\fh} \in \mathbb{R}^{\bar{n}} :\, \mathbf{c}_i^{\fh} = (c_{i,t}^{\fh} )_{t=1}^{\bar{n}}, \, \fc_i \in \mathsf{C}_{\text{\tiny conv}}, \forall i \in [\![M]\!] \big\}.$
		\item $\alpha_n \triangleq \sqrt{\bar{n}\epsilon_n'}$ where $\epsilon_n' = a^2/\bar{n}^{2(1+b)}$ with $b>0$ being an arbitrarily small constant.
		\item $\fv_n = \left(-A_{\max} + 2\alpha_n\right)\mathbf{1}_n,$ where $\mathbf{1}_n = (1,\ldots,1)^\top \in \mathbb{R}^n.$
	\end{itemize}
	\begin{customlemma}{6}
		\label{Lem.Converse}
		Let $R$ be an achievable identification rate. Consider a sequence of $(n,\allowbreak M(n\allowbreak,R),\allowbreak K(n,\allowbreak \kappa), \allowbreak \lambda_1^{(n)}, \allowbreak \lambda_2^{(n)})$-codes $(\mathsf{C}_{\text{\tiny conv}}^{(n)},\allowbreak \D^{(n)})$ with $K(n,\kappa)=n^{\kappa}$ with $\kappa \in [0,1),$ such that $\lambda_1^{(n)}$ and $\lambda_2^{(n)}$ tend to zero as $n \rightarrow \infty.$ Then, given a sufficiently large $n,$ the convoluted codebook $\mathsf{C}_{\text{\tiny conv}}^{\fh}$ meets the property: Every pair of codewords $\fc_{i_1}^{\fh}$ and $\fc_{i_2}^{\fh},$ inside $\mathsf{C}_{\text{\tiny conv}}^{\fh},$ with $i_1,i_2 \in [\![M]\!]$ and $i_1 \neq i_2,$ are separated by\vspace*{-.2mm}
		\begin{align}
			\label{Ineq.Conv_Distance}
			\big\| \fc_{i_1}^{\fh} - \fc_{i_2}^{\fh} \big\| \geq 2\sqrt{\bar{n}\epsilon_n'} \triangleq 2\alpha_n .
		\end{align}
	\end{customlemma}
		\begin{proof}
			The proof is provided in Appendix \ref{App.Converse_Proof}. The key step is to exploit the continuity of the Gaussian output distribution to show that, if \eqref{Ineq.Conv_Distance} is violated, then sum of the two error probabilities converges to one, contradicting identifiability.
		\end{proof}
	
	Next, we use Lemma~\ref{Lem.Converse} to prove the upper bound on the identification capacity. Observe that since the minimum distance of the convoluted codebook is $2\alpha_n,$ we can arrange non-overlapping spheres $\S_{\fc_i^{\fh}}(n,\alpha_n)$ whose centers belong to the convoluted codebook $\mathsf{C}_{\text{\tiny conv}}^{\fh}.$ We note that in general the spheres belonging to such packing arrangement are inscribed in the hypercube $\mathsf{Q}_{\fv_n}(\bar{n},A_{\rm max} + 2\alpha_n).$ On the other hand such sphere packing arrangement is in general not saturated. This phenomenon occurs since we did not arrange a sphere packing from the beginning inside the convoluted codebook but rather enumerate the original codewords by associating a hypersphere inside the induced space, i.e., $\mathsf{C}_{\text{\tiny conv}}^{\fh}.$ However, the number of codewords, i.e., the number of codewords derived out of such induced sphere packing, $M,$ is upper bounded by the number of codewords derived out of a saturated packing arrangement, $\bar{M},$ accomplished directly in $\mathsf{C}_{\text{\tiny conv}}^{\fh}.$ Therefore, it follows that $M,$ is bounded by
	\begin{align}
		\label{Ineq.Codebook_Size_UB}
		M & = \frac{\text{Vol}\left(\bigcup_{i=1}^{M} \S_{\fc_i}^{\fh}(\bar{n},\alpha_n) \right)}{\text{Vol}(\S_{\fc_1}^{\fh}(\bar{n},\alpha_n))} \stackrel{(a)}{\leq} \frac{\text{Vol}\left(\bigcup_{i=1}^{\bar{M}} \S_{\fc_i}^{\fh}(\bar{n},\alpha_n) \right)}{\text{Vol}(\S_{\fc_1}^{\fh}(\bar{n},\alpha_n))}
		\nonumber\\&
		\stackrel{(b)}{=} \frac{\Updelta_n(\mathscr{S}) \cdot \text{Vol}\big( \mathsf{Q}_{\fv_n}(\bar{n},A_{\rm max} + 2\alpha_n) \big)}{\text{Vol}(\S_{\fc_1}^{\fh}(\bar{n},\alpha_n))}
		\nonumber\\&
		\stackrel{(c)}{\leq} 2^{-0.599\bar{n}} \cdot \frac{(A_{\rm max}+2\alpha_n)^{\bar{n}}}{\text{Vol}(\S_{\fc_1}^{\fh}(\bar{n},\alpha_n))},
	\end{align}
	where $(a)$ holds since a saturated packing encompass the maximum possible number of sphere, $(b)$ conforms the density definition and $(c)$ exploits  \eqref{Ineq.Density} and the following:
	\begin{align}
		\mathsf{C}_{\text{\tiny conv}}^{\fh} & \subseteq \mathsf{Q}_{\fv_n}(\bar{n},A_{\rm max} + 2\alpha_n)
		\nonumber\\&
		= \big\{ \fc_i^{\fh} \in \mathbb{R}^{\bar{n}} \hspace{-.4mm}:\hspace{-.2mm} - (A_{\rm max} + \alpha_n) \leq c_{i,t}^{\fh} \leq A_{\rm max} + \alpha_n \big\},
	\end{align}
	which implies 
	\begin{equation}
		\text{Vol}( \mathsf{C}_{\text{\tiny conv}}^{\fh} ) \leq \text{Vol}( \mathsf{Q}_{\fv_n}(\bar{n},A_{\rm max} + 2\alpha_n)) = (A_{\rm max} + 2\alpha_n)^{\bar{n}}.
	\end{equation}
	Thereby, we obtain the following upper bound on the logarithm of the codebook size
	\begin{align}
		\label{Ineq.Log_M_UB}
			 \log M \leq \bar{n} \log (A_{\rm max} + 2\alpha_n) - \bar{n} \log \alpha_n + \frac{1}{2} \bar{n} \log \bar{n} + \mathcal{O}(\bar{n}),
	\end{align}
	Now, for $\alpha_n = \sqrt{\bar{n}\epsilon_n'} = a/\bar{n}^{\frac{1+2b}{2}}$ and $A_{\rm max} = KLP_{\rm max} = \mathcal{O}(n^{\kappa}),$ we obtain
	\begin{align}
		\label{Ineq.Log_M_UB2}
		\log M & \leq \bar{n} \log (A_{\rm max}) + \bar{n} \log (1 + 2\alpha_n/A_{\rm max})
		\nonumber\\&
		+ \Big( \frac{1+2b}{2} \Big) \bar{n} \log \bar{n} - \bar{n} \log a \sqrt{\pi} + \frac{1}{2} \bar{n} \log \bar{n} + \mathcal{O}(\bar{n})
		\nonumber\\&
		= \bar{n} \log ( KLP_{\rm max} ) + \bar{n} \log \Big(1 + \frac{2a}{LP_{\rm max}n^{\kappa}\bar{n}^{\frac{1+2b}{2}}} \Big)
		\nonumber\\&
		- \bar{n} \log a\sqrt{\pi} + \frac{1}{2} \bar{n} \log \bar{n} + \mathcal{O}(\bar{n})
		\nonumber\\&
		= \kappa \bar{n} \log n + \bar{n} \log LP_{\rm max} + o(1) + \Big( \frac{2+2b}{2} \Big) \bar{n} \log \bar{n}
		\nonumber\\&
		- \bar{n} \log a\sqrt{\pi} + \mathcal{O}(\bar{n}),
	\end{align}
	where the leading-order term scales as $\bar{n} \log \bar{n} \sim n \log n,$ we set $M=2^{(n \log n)R},$ and we obtain
	\begin{equation*}
		\resizebox{.492\textwidth}{!}{
			$\begin{aligned}
				\hspace{-1.2mm} R \leq \frac{1}{n \log n} \Big[ \Big( \frac{2 + 2\kappa + 2b}{2} \Big) \, n \log n + \bar{n} \log \Big( \frac{LP_{\rm max}}{a\sqrt{\pi}} \Big)  + \mathcal{O}(n) \Big]
			\end{aligned}$}
	\end{equation*}
	which tends to $1 + \kappa + b$ as $n \to \infty.$ Now, since $b>0$ is arbitrarily small, an achievable rate must satisfy $R \leq 1 + \kappa.$ This completes the converse proof of Theorem~\ref{Th.ISI-Gauss-Cap}.
	
	\section{Conclusions and Outlook}
	\label{Sec.Conclusion}
	This work studies the identification over ISI Gaussian channels, generalizing the ISI-free model in \cite{Salariseddigh_ITW} motivated by wireless and molecular communications. We show that reliable identification is achievable with a super-exponential codebook $M=2^{(n\log n)R}$ when the DTFT spectrum is bounded away from zero, even for sub-linear ISI depth $K$. We establish lower and upper bounds on the identification rate $R$ in terms of the ISI rate $\kappa$. The gap between the derived bounds reflects the limitations of the respective achievability and converse techniques: the sphere-packing construction and distance decoder on the achievability side, and the relaxation to pairwise separation in the converse. On the achievability argument, the proposed construction need not exhaust all codes that could achieve the identification capacity. The converse extends to average-power constraints by replacing the peak-power region with the corresponding energy-constrained region. Since the converse relies only on pairwise output separation, the same packing argument applies after bounding the admissible region's volume. The converse does not require $\inf |H(\omega)|>0$ and therefore remains valid for frequency-selective ISI channels with arbitrarily small spectral magnitudes or spectral nulls, since the required separation is established directly for the convoluted codewords by contradiction with reliability. Future work includes finite-blocklength analysis, complex-valued and time-varying ISI channels, colored noise, practical code constructions, and ISI channels with larger or non-vanishing ISI depth, i.e., $\kappa \in [1/2,1]$.

	\section{Acknowledgments}
 	The authors acknowledge Professor Dr. -Ing. Dr. rer. nat. Holger Boche (Technical University of Munich) for helpful discussions. Dr. -Math. Christian Deppe was supported by the Federal Ministry of Research, Technology and Space (BMFTR) under the research program Communication Systems – “Souverän. Digital. Vernetzt.” through the joint project 6G-life (Project ID: 16KIS2415) and the project QSTARS (Project ID: 16KIS2196).

	\appendices
	
	\section{Proof of Lemma \ref{Lem.Constant_K}; Minimum Spectral Value of CIR with Constant Depth}
	\label{App.LB_FS}
	
	In the following, we employ the Z-transform \cite{Oppenheim99} and standard techniques to show that the minimal value of spectrum for any CIR vector whose length $K \geq 1$ is a fixed constant, i.e., $\kappa=0$ or $\kappa \to 0$ as $n \to \infty$ so that $K(n,\kappa) = \mathcal{O}(1),$ remains strictly positive.
	\begin{proof}
	Consider the DTFT of $\fh,$ i.e.,
	\begin{align}
		\label{Eq.DTFT}
		H(\omega) = \sum_{k=0}^{K-1} h_k e^{-j\omega k}.
	\end{align}
	If we let $z=e^{j\omega},$ the DTFT in \eqref{Eq.DTFT} is a polynomial in $z^{-1},$ that is, we have $$H(z) = h_0 + h_1z^{-1} + h_2z^{-2} + \ldots +h_{K-1}z^{-(K-1)}.$$ Now, multiplying by $z^{K-1}$ we obtain a standard polynomial $H(z)$ of degree $K-1$ as follow: $z^{K-1} H(z) = h_0z^{K-1} + h_1 z^{K-2} + \cdots + h_{K-1}.$ Now, observe that according to the fundamental theorem of algebra \cite{Ahlfors79}, any polynomial of degree $K-1$ has exactly $K-1$ complex roots, $z_0,\ldots,z_{K-2}.$ Thereby, we have $$z^{K-2} H(z) = h_0 \prod_{k=0}^{K-2} (z - z_k).$$ Now, we take the magnitude from both side and evaluate on the unit circle, i.e., on $|z|=1,$
	\begin{align}
		|z^{K-2}| \cdot |H(e^{j\omega})| = |h_0| \prod_{k=1}^{K-2} |e^{j\omega} - z_k|,
	\end{align}
	since $|z^{K-2}| = |e^{j(K-2)\omega}|,$ we obtain the DTFT magnitude as follows
	\begin{align}
		\label{Eq.Mag_DTFT}
		|H(e^{j\omega})| = |h_0| \prod_{k=0}^{K-2} |e^{j\omega} - z_k|.
	\end{align}
	Observe that since the DTFT has no zeros on the unit circle, that is, $e^{j\omega} \neq z_k,\, \forall k \in [\![K-1]\!]$ where $\omega \in [0,2\pi].$ The magnitude at any frequency is the product of the distances from the point on the unit circle $e^{j\omega}$ to each of the fixed roots in the complex plane. Because coefficients $h_k$ and $K$ are fixed, the positions of these roots are locked. There is a strictly positive minimum distance between the unit circle and the nearest root. Therefore, the product of these distances can never be zero and cannot decay because the geometry is static. More specifically, let define $d_k(\omega) \triangleq |e^{j\omega} - z_k|$ and $d_k^{\rm min} \triangleq \min_{\omega} d_k(\omega)$ be the minimum distance from the unit circle to the $k$-th root. Since $z_k$ is not on the unit circle, $d_k>0,\, \forall k \in [\![K-1]\!].$ Thereby, the non-decaying minimum magnitude reads
	\begin{align}
		H(e^{j\omega}) = |h_0| \prod_{k=0}^{K-2} d_k(\omega) \geq |h_0| \prod_{k=0}^{K-2} d_k^{\rm min} \triangleq H_{\rm LB} > 0 ,
	\end{align}
	where $H_{\rm LB}$ is a uniform positive lower bound on $H(e^{j\omega}).$
	Thereby,
	\begin{equation}
		H_{\rm min} \triangleq \frac{1}{\sqrt{2\pi}} \cdot \inf_{\omega \in [-\pi,\pi]} |H(\omega)| \geq \frac{1}{\sqrt{2\pi}} \cdot H_{\rm LB}  > 0 .
	\end{equation}
	\end{proof}

	\section{Proof of Lemma~\ref{Lem.Increasing_K}; Minimum Spectral Value of CIR with Increasing Depth}
	\label{App.Increasing_K}

	\begin{proof}
		Let $\mathbf{d}_{ij}\triangleq \mathbf{c}_i-\mathbf{c}_j.$ Since $\mathbf{c}_i^{\mathbf h}=\mathbf{H}\mathbf{c}_i$ and $\mathbf{c}_j^{\mathbf h}=\mathbf{H}\mathbf{c}_j,$ we have $\mathbf{c}_i^{\mathbf h}-\mathbf{c}_j^{\mathbf h}
		=
		\mathbf{H}(\mathbf{c}_i-\mathbf{c}_j)
		=
		\mathbf{H}\mathbf{d}_{ij}.$ Therefore,
		\begin{equation}
			\|\mathbf{c}_i^{\mathbf h}-\mathbf{c}_j^{\mathbf h}\|^2
			=
			\|\mathbf{H}\mathbf{d}_{ij}\|^2
			=
			\mathbf{d}_{ij}^{T}\mathbf{H}^{T}\mathbf{H}\mathbf{d}_{ij}.
			\label{eq:norm_quadratic}
		\end{equation}
		For $\mathbf{d}_{ij}\neq\mathbf{0}$, the Rayleigh quotient gives
		\begin{equation}
			\frac{
				\mathbf{d}_{ij}^{T}\mathbf{H}^{T}\mathbf{H}\mathbf{d}_{ij}
			}{
				\mathbf{d}_{ij}^{T}\mathbf{d}_{ij}
			}
			\geq
			\lambda_{\min}(\mathbf{H}^{T}\mathbf{H}).
			\label{eq:rayleigh_bound}
		\end{equation}
		Since the eigenvalues of $\mathbf{H}^{T}\mathbf{H}$ are the squared
		singular values of $\mathbf{H}$, $\lambda_{\min}\!(\mathbf{H}^{T}\mathbf{H}) = \sigma_{\min}^{2}(\mathbf{H}).$
		Thereby, $\|\mathbf{c}_i^{\mathbf h}-\mathbf{c}_j^{\mathbf h}\|^2
		\geq
		\sigma_{\min}^{2}(\mathbf{H})
		\|\mathbf{c}_i-\mathbf{c}_j\|^2.$
		Taking square roots gives
		\begin{equation}
			\label{Eq.LB_Conv_Dist}
			\|\mathbf{c}_i^{\mathbf h}-\mathbf{c}_j^{\mathbf h}\| \geq \sigma_{\min}(\mathbf{H}) \left\|\mathbf{c}_i-\mathbf{c}_j\right\|.
		\end{equation}
		
		Next, to simplify the lower bound in \eqref{Eq.LB_Conv_Dist} we provide the following lemma which will be useful in establishing a lower bound on the distance between the transformed vectors. In particular, it allows us to relate the minimum singular value of the full convolution matrix $\mathbf{H}$ to that of a specific square submatrix of $\mathbf{H}.$
		\begin{customlemma}{4}[Monotonicity of Singular Values Under Row Augmentation]
			\label{Lem.Mon_Sing_Values}
			Let $\mathbf H\in\mathbb{R}^{m\times n}$, with $m>n$, be a full-column-rank convolution matrix, and let $\mathbf H_s\in\mathbb{R}^{n\times n}$ be any square submatrix obtained by selecting $n$ rows of $\mathbf H$ (e.g., the first $n$ rows). Then for every $i \in [\![n]\!]$
			\begin{equation}\label{eq:HsH}
				\sigma_i(\mathbf H_s)\leq \sigma_i(\mathbf H) .
			\end{equation}
			In particular, $\sigma_{\min}(\mathbf H_s) \leq \sigma_{\min}(\mathbf H).$ Consequently, any lower bound on $\sigma_{\min}(\mathbf H_s)$ is also a valid lower bound on $\sigma_{\min}(\mathbf H)$.
		\end{customlemma}
		\begin{proof}
			The proof is provided in Appendix \ref{App.Mon_Sing_Values}.
		\end{proof}
		Thereby, employing Lemma \ref{Lem.Mon_Sing_Values} we obtain $\sigma_{\min}(\mathbf{H}_s) \leq \sigma_{\min}(\mathbf{H}).$
		Consequently,
		\begin{equation}
			\|\mathbf{c}_i^{\mathbf h}-\mathbf{c}_j^{\mathbf h}\|
			\geq
			\sigma_{\min}(\mathbf{H}_s)
			\|\mathbf{c}_i-\mathbf{c}_j\|.
			\label{eq:distance_Hs}
		\end{equation}
		Thus, any lower bound established for
		$\sigma_{\min}(\mathbf{H}_s)$ automatically yields the same lower
		bound for the channel-induced codeword distances associated with the
		full tall convolution matrix $\mathbf{H}$. Having established the bound
		\[
		\|\mathbf{c}_i^{\mathbf h}-\mathbf{c}_j^{\mathbf h}\|
		\geq
		\sigma_{\min}(\mathbf{H}_s)
		\|\mathbf{c}_i-\mathbf{c}_j\|,
		\]
		we next derive a lower bound on $\sigma_{\min}(\mathbf{H}_s)$. Specifically, we verify that $\mathbf{H}_s$ satisfies the SDD conditions of the following lemma. This lemma lower-bounds the smallest singular value of an SDD matrix in terms of its row and column diagonal-dominance margins. Thus, showing that the corresponding margins of $\mathbf{H}_s$ are positive directly yields a positive lower bound on $\sigma_{\min}(\mathbf{H}_s)$ and, consequently, on the distance between the transformed vectors.

		\begin{customlemma}{5}[Minimum Singular Value Bound for SDD Matrices]
			\label{Lem.LB_Sing}
			Assume matrix $\fA = [a_{i,j}] \in \mathbb{R}^{n \times n}$ and $\fA^T$ are both SDD, i.e., $|a_{i,i}| - \sum_{\substack{j=1\\j \neq i}}^n |a_{i,j}| > 0,\, \forall i \in [\![n]\!],$ and $|a_{i,i}| - \sum_{\substack{j=1\\j \neq i}}^n |a_{j,i}| > 0,\, \forall i \in [\![n]\!].$ Moreover, let define the diagonal dominance margins as follows
			\begin{align}
				L_1 \hspace{-.3mm} \triangleq \hspace{-.3mm} \min_{i} \Big\{ |a_{i,i}| - \sum_{\substack{j=1\\j \neq i}}^n |a_{i,j}| \Big\}, L_2 \hspace{-.3mm} \triangleq \hspace{-.3mm} \min_{i} \Big\{ |a_{i,i}| - \sum_{\substack{j=1\\j \neq i}}^n |a_{j,i}| \Big\} .
			\end{align}
			Then, the smallest singular value of $\fA,$ $\sigma_n(\fA),$ satisfies
			\begin{align}
				\sigma_n(\fA) \geq \sqrt{L_1L_2} > 0.
			\end{align}
			\begin{proof}
				The proof is provided in Appendix \ref{App.LB_Sing}.
			\end{proof}
		\end{customlemma}
		
		Next, we observe that, for the finite-length convolution matrix $\mathbf H_{\bar n\times n}$, the associated $n\times n$ square lower-triangular Toeplitz convolution matrix $\mathbf H_s$, formed by its first $n$ rows, has all diagonal entries equal to $h_0$. Since every diagonal entry of $\mathbf H_s$ is equal to $h_0$, while the sum of the magnitudes of the off-diagonal entries in each row and each column is bounded above by $\sum_{k=1}^{K-1}|h_k|$, it follows that for every row \(i\), \begin{equation}
			\sum_{j\neq i}\left|(\mathbf H_s)_{ij}\right| \leq \sum_{k=1}^{K-1}|h_k| < |h_0| = \left|(\mathbf H_s)_{ii}\right|.
		\end{equation}
		Similarly, for every column \(j\), \begin{equation} \sum_{i\neq j}\left|(\mathbf H_s)_{ij}\right| \leq \sum_{k=1}^{K-1}|h_k| < |h_0| = \left|(\mathbf H_s)_{jj}\right|. \end{equation} Hence, the square lower-triangular Toeplitz convolution matrix $\mathbf H_s$ is SDD by both rows and columns which allows Lemma~\ref{Lem.LB_Sing} to be applied to $\mathbf H_s$. Moreover, the diagonal dominance margins $L_1$ and $L_2$ satisfies
		\begin{align}
			L_1 = L_2 \triangleq |h_0| - \sum_{k=1}^{K-1} |h_k| \triangleq m > 0,
		\end{align}
		for $m$ being independent of $n$, then leveraging Lemma~\ref{Lem.LB_Sing} with setting $\fA = \fH_s$ guarantees that
		\begin{equation}
			\label{Eq.Sig_Min_H_s}
			\sigma_{\min}(\mathbf H_s)\ge\sqrt{L_1L_2} = \sqrt{m^2} = m > 0.
		\end{equation}
		Therefore, whenever the leading channel tap dominates the sum of the magnitudes of the remaining ISI taps, the smallest singular value of the tall convolution matrix admits the explicit lower bound in \eqref{Eq.Sig_Min_H_s}. This bound is independent of the output length $\bar{n}=n+K-1$ and depends only on the channel coefficients.
	\end{proof}

	\section{Proof of Lemma~\ref{Lem.Mon_Sing_Values}; Monotonicity of Singular Values Under Row Augmentation}
	\label{App.Mon_Sing_Values}
	
	\begin{proof}
		Let $\mathbf H\in\mathbb R^{m\times n}$ denote the full tall convolution
		matrix, and let $\mathbf H_s\in\mathbb R^{n\times n}$ be obtained by selecting
		$n$ rows of $\mathbf H$. Let
		$\mathbf P\in\mathbb R^{n\times m}$ denote the corresponding row-selection
		matrix. Then
		\begin{equation}
			\mathbf H_s=\mathbf P\mathbf H.
		\end{equation}
		Since $\mathbf P$ selects distinct rows, its rows are orthonormal, and hence
		\begin{equation}
			\mathbf P\mathbf P^T=\mathbf I_n.
		\end{equation}
		Consequently, $\mathbf P^T\mathbf P$ is the orthogonal projector onto the
		subspace spanned by the selected coordinate directions. In particular,
		\begin{equation}\label{eq:projection_order}
			0\preceq\mathbf P^T\mathbf P\preceq\mathbf I_m.
		\end{equation}
		Thus, $\mathbf I_m-\mathbf P^T\mathbf P\succeq0.$ Since congruence preserves positive semidefiniteness, it follows that
		\begin{equation}
			\mathbf H^T
			\bigl(\mathbf I_m-\mathbf P^T\mathbf P\bigr)
			\mathbf H
			\succeq0.
		\end{equation}
		Equivalently,
		\begin{equation}\label{eq:gram_loewner}
			\mathbf H^T\mathbf H
			-
			\mathbf H^T\mathbf P^T\mathbf P\mathbf H
			\succeq0.
		\end{equation}
		Using $\mathbf H_s=\mathbf P\mathbf H$, we obtain
		\begin{equation}\label{eq:gram_inequality}
			\mathbf H_s^T\mathbf H_s
			=
			\mathbf H^T\mathbf P^T\mathbf P\mathbf H
			\preceq
			\mathbf H^T\mathbf H.
		\end{equation}
		Equivalently, for every $\mathbf x\in\mathbb R^n$,
		\begin{equation}
			\mathbf x^T\mathbf H_s^T\mathbf H_s\mathbf x
			=
			\|\mathbf H_s\mathbf x\|_2^2
			\leq
			\|\mathbf H\mathbf x\|_2^2
			=
			\mathbf x^T\mathbf H^T\mathbf H\mathbf x.
		\end{equation}
		Thus, the row selection cannot increase the quadratic form associated with
		the Gram matrix. 
		
		We next use the spectral monotonicity induced by the Loewner order. Let the
		eigenvalues of a symmetric matrix $\fA\in\mathbb R^{n\times n}$ be ordered
		nonincreasingly as
		\begin{equation}
			\lambda_1(\fA)
			\geq
			\lambda_2(\fA)
			\geq
			\cdots
			\geq
			\lambda_n(\fA).
		\end{equation}
		Weyl's monotonicity theorem \cite{HornJohnson2013} for Hermitian matrices states that
		\begin{equation}
			\fA\preceq\fB
			\quad\Longrightarrow\quad
			\lambda_i(\fA)
			\leq
			\lambda_i(\fB),
			\qquad i=1,\ldots,n.
		\end{equation}
		Applying this result to \eqref{eq:gram_inequality}, with
		\begin{equation}
			\fA=\mathbf H_s^T\mathbf H_s,
			\qquad
			\fB=\mathbf H^T\mathbf H,
		\end{equation}
		gives
		\begin{equation}\label{eq:weyl_monotonicity}
			\lambda_i(\mathbf H_s^T\mathbf H_s)
			\leq
			\lambda_i(\mathbf H^T\mathbf H),
			\qquad i=1,\ldots,n.
		\end{equation}
		Since $\mathbf H_s^T\mathbf H_s$ and $\mathbf H^T\mathbf H$ are Gram
		matrices, their eigenvalues are the squared singular values of
		$\mathbf H_s$ and $\mathbf H$, respectively. Hence,
		\begin{equation}
			\lambda_i(\mathbf H_s^T\mathbf H_s)
			=
			\sigma_i^2(\mathbf H_s),
			\qquad
			\lambda_i(\mathbf H^T\mathbf H)
			=
			\sigma_i^2(\mathbf H).
		\end{equation}
		Substituting these identities into \eqref{eq:weyl_monotonicity} yields
		\begin{equation}
			\sigma_i^2(\mathbf H_s)
			\leq
			\sigma_i^2(\mathbf H),
			\qquad i=1,\ldots,n.
		\end{equation}
		Since singular values are nonnegative, taking square roots gives
		\begin{equation}\label{eq:HsH}
			\sigma_i(\mathbf H_s)
			\leq
			\sigma_i(\mathbf H),
			\qquad i=1,\ldots,n.
		\end{equation}
		In particular, setting $i=n$ yields
		\begin{equation}\label{eq:sigmaminHs}
			\sigma_{\min}(\mathbf H_s)
			\leq
			\sigma_{\min}(\mathbf H).
		\end{equation}
		This completes the proof.
	\end{proof}

	\section{Proof of Lemma~\ref{Lem.LB_Sing}; Minimum Singular Value Bound for SDD Matrices}
	\label{App.LB_Sing}
	
	\begin{proof}
		Since $\mathbf A$ is strictly diagonally dominant (SDD) by rows, the Levy--Desplanques theorem \cite{HornJohnson2013} guarantees that $\mathbf A$ is non-singular. Moreover, by Varah's bound \cite{Varah75} for the inverse of a SDD matrix, $\|\mathbf A^{-1}\|_{\infty} \le L_1^{-1}$ where
		\begin{equation}
			L_1 \triangleq \min_{i \in [\![n]\!]} \Big\{ |a_{i,i}| - \sum_{\substack{j=1\\j \neq i}}^n |a_{i,j}| \Big\} .
		\end{equation}
		Similarly, since $\mathbf A^T$ is also SDD by rows, another application of Varah's bound yields $\|(\mathbf A^T)^{-1}\|_{\infty} \le L_2^{-1}$
		where
		\begin{equation}
			L_2 \triangleq \min_{i \in [\![n]\!]} \Big\{ |a_{i,i}| - \sum_{\substack{j=1\\j \neq i}}^n |a_{j,i}| \Big\} .
		\end{equation}
		Next, using the identity $(\mathbf A^T)^{-1} = (\mathbf A^{-1})^T,$
		together with the duality relation between the induced $1$- and $\infty$-norms, $\|\mathbf X\|_1=\|\mathbf X^T\|_{\infty},$ we obtain
		\begin{equation}
			\|\mathbf A^{-1}\|_1
			=
			\|(\mathbf A^{-1})^T\|_{\infty}
			=
			\|(\mathbf A^T)^{-1}\|_{\infty}
			\le L_2^{-1}
		\end{equation}
		Consequently,
		\begin{equation}
			\|\mathbf A^{-1}\|_{\infty}
			\le
			L_1^{-1}
			\qquad
			\|\mathbf A^{-1}\|_1
			\le
			L_2^{-1}
		\end{equation}
		Now apply the standard inequality relating the induced matrix norms,
		\begin{equation}
			\|\mathbf X\|_2 \le	\sqrt{\|\mathbf X\|_1\,\|\mathbf X\|_{\infty}},
		\end{equation}
		with $\mathbf X=\mathbf A^{-1}$. It follows that
		\begin{align}
			\|\mathbf A^{-1}\|_2 \le \sqrt{\|\mathbf A^{-1}\|_1 \|\mathbf A^{-1}\|_{\infty}}\le \sqrt{L_2^{-1} L_1^{-1}} = (L_1L_2)^{-1/2}
		\end{align}
		Finally, since $\mathbf A$ is nonsingular, its smallest singular value satisfies $\sigma_n(\mathbf A) =\|\mathbf A^{-1}\|_2^{-1}.$
		Therefore,
		\begin{equation}
			\sigma_n(\mathbf A) = \|\mathbf A^{-1}\|_2^{-1} \ge \sqrt{L_1L_2}.
		\end{equation}
		which completes the proof.
	\end{proof}

	\section{Proof of Lemma~\ref{Lem.Converse}; Pairwise Separation of Achievable Convoluted Codebooks}

	\label{App.Converse_Proof}
	
	We prove Lemma~\ref{Lem.Converse} by contradiction. Specifically, we assume that the condition in \eqref{Ineq.Conv_Distance} does not hold and demonstrate that this assumption leads to a contradiction. In particular, we show that the sum of the type I and type II error probabilities tends to one, i.e., $\lim_{n \to \infty} \big[ P_{e,1}(i_1) + P_{e,2}(i_2,i_1) \big] = 1.$
	
	Fix $e_1$ and $e_2$. Let $\tau,\theta,\zeta>0$ be arbitrarily small. Assume to the contrary that 
	there exist two messages $i_1$ and $i_2$, where $i_1\neq i_2$, such that
	\begin{align}
		\label{Eq.Alpha_nFast}
		\big\| \fc_{i_1}^{\fh} - \fc_{i_2}^{\fh} \big\| < \sqrt{\bar{n}\epsilon_n'} \triangleq \alpha_n = a/\bar{n}^{\frac{1+2b}{2}}.
	\end{align}
	Now let us define two subsets as follows
	\begin{align}
		\label{Eq.Event_BC}
		\mathsf{D}_{i_1,i_2} & = \Big\{ \fy \in \mathsf{D}_{i_1}:
		\| \fy - \fc_{i_2}^{\fh}\| \leq \sqrt{\bar{n}( \sigma_Z^2 + \zeta)} \Big\}
		\nonumber\\
		\mathsf{E}_{i_2} & = \Big\{ \fy \in \mathbb{R}^{\bar{n}}:
		\| \fy - \fc_{i_2}^{\fh} \| \leq \sqrt{\bar{n}(\sigma_Z^2+\zeta)} \Big\} .
	\end{align}
	Next, we can bound the type I error probability according to the events degined in \eqref{Eq.Event_BC} as follows
	\begin{align}
		& 1-P_{e,1}(i_1)
		\nonumber\\
		& = \int_{\mathsf{D}_{i_1}} \hspace{-2mm} f_{\fZ}(\fy - \fc_{i_1}^{\fh}) d\fy
		\nonumber\\
		& = \int_{\mathsf{D}_{i_1,i_2}} f_{\fZ}(\fy - \fc_{i_1}^{\fh}) d\fy + \int_{\mathsf{D}_{i_1} \setminus \mathsf{D}_{i_1,i_2}} \hspace{-2mm} f_{\fZ}(\fy - \fc_{i_1}^{\fh}) d\fy 
		\nonumber\\&
		\leq \int_{\mathsf{D}_{i_1,i_2}} f_{\fZ}(\fy - \fc_{i_1}^{\fh}) d\fy + \int_{\mathsf{E}_{i_2}^c} f_{\fZ}(\fy - \fc_{i_1}^{\fh}) d\fy ,
		\label{Eq.Pe1boundConv0Fast}
	\end{align}
	where the last inequality holds since $\mathsf{D}_{i_1} \setminus \mathsf{D}_{i_1,i_2} \subset \mathsf{E}_{i_2}^c.$
	Consider the second integral, for which the domain is $\mathsf{E}_{i_2}^c$. Then, by the triangle inequality
	\begin{align}
		\| \fy - \fc_{i_1}^{\fh} \| & \geq \| \fy - \fc_{i_2}^{\fh} \| - \| \fc_{i_1}^{\fh} - \fc_{i_2}^{\fh} \|
		\nonumber\\&
		> \sqrt{\bar{n}(\sigma_Z^2+\zeta)} - \| \fc_{i_1}^{\fh} - \fc_{i_2}^{\fh} \|
		\nonumber\\&
		\geq\sqrt{\bar{n}(\sigma_Z^2+\zeta)} - \alpha_n,
	\end{align}
	For sufficiently large $n$, this implies the following subset
	\begin{align}
		\mathsf{F}_{i_1,i_2}^c = \Big\{\fy \in \mathbb{R}^{\bar{n}} \; : \, \| \fy - \fc_{i_1}^{\fh}\| > \sqrt{\bar{n}(\sigma_Z^2 + \eta)} \Big\},
		\label{Eq.Regiong0}
	\end{align}
	for $\eta<\frac{\zeta}{2}.$ That is, the event
	\begin{equation}
		\big\{\fy \in \mathbb{R}^{\bar{n}} \; : \, \| \fy - \fc_{i_2}^{\fh} \| \geq \sqrt{\bar{n}(\sigma_Z^2+\zeta)} \big\}
	\end{equation}
	implies
	\begin{equation}
		\big\{\fy \in \mathbb{R}^{\bar{n}} \; : \, \| \fy - \fc_{i_1}^{\fh} \| \geq \sqrt{\bar{n}(\sigma_Z^2+\eta)} \big\}.
	\end{equation}
	Thereby, we conclude that $\mathsf{F}_{i_1,i_2}^c \supset \mathsf{E}_{i_2}^c.$ Hence, the second integral in \eqref{Eq.Pe1boundConv0Fast} is bounded by
	\begin{align}
		\label{Ineq.F_i_1_i_2_compl}
		\int_{\mathsf{F}_{i_1,i_2}^c} f_{\fZ}( \fy - \fc_{i_1}^{\fh}) d\fy & = \Pr \Big( \| \fY - \fc_{i_1}^{\fh} \| > \sqrt{\bar{n}(\sigma_Z^2 + \eta)} \Big)
		\nonumber\\&
		\stackrel{(a)}{=} \Pr \big( \norm{\fZ}^2 \hspace{-.5mm} - \hspace{-.5mm} \bar{n}\sigma_Z^2 > \bar{n}\eta \big)
		\hspace{-.5mm} \stackrel{(b)}{\leq} \hspace{-.5mm} \frac{3\sigma_Z^4}{n\eta^2} \hspace{-.5mm} \leq \hspace{-.5mm} \tau,
	\end{align}
	for sufficiently large $n$, where $(a)$ holds by the Chebyshev's inequality, followed by the substitution of $\fZ \equiv \fY - \fc_{i_1}^{\fh}$ and $(b)$ uses $n \leq \bar{n}.$ Thus, merging \eqref{Eq.Pe1boundConv0Fast} and \eqref{Ineq.F_i_1_i_2_compl} and gives
	\begin{align}
		\label{Eq.ComplTypeIFast}
		1 - \tau - P_{e,1}(i_1) \leq \int_{\mathsf{D}_{i_1,i_2}} f_{\fZ}(\fy - \fc_{i_1}^{\fh}) d\fy.
	\end{align}
	Now, we can focus on the inner integral with domain of $\mathsf{D}_{i_1,i_2}$, i.e., when
	\begin{align}
		\| \fy - \fc_{i_2}^{\fh} \| \leq \sqrt{\bar{n}(\sigma_Z^2+\zeta)}.
		\label{Eq.ui2DistFast}
	\end{align}
	Observe that
	\begin{align}
		\label{Ineq.Error_Diff}
		& \big| f_{\fZ}(\fy - \fc_{i_1}^{\fh}) - f_{\fZ}(\fy - \fc_{i_2}^{\fh}) \big|
		\nonumber\\& 
		= f_{\fZ}(\fy - \fc_{i_1}^{\fh}) \cdot \Big| 1 - \exp\Big( -\frac{\Delta_{i_1,i_2}(\fy)}{2\sigma_Z^2} \Big) \Big| 
	\end{align}
	where $\Delta_{i_1,i_2}(\fy) \triangleq -\big( \| \fy - \fc_{i_2}^{\fh} \|^2 - \| \fy - \fc_{i_1}^{\fh} \|^2 \big) / 2\sigma_Z^2 .$
	By the triangle inequality,
	\begin{align}
		\label{Ineq.triangleFast}
		\| \fy - \fc_{i_1}^{\fh} \| \leq \| \fy -\fc_{i_2}^{\fh} \| + \| \fc_{i_1}^{\fh} - \fc_{i_2}^{\fh} \|.
	\end{align}
	Taking the square of both sides, we have
	\begin{align}
		& \| \fy - \fc_{i_1}^{\fh} \|^2
		\nonumber\\&
		\leq \| \fy - \fc_{i_2}^{\fh} \|^2 + \| \fc_{i_1}^{\fh} - \fc_{i_2}^{\fh} \|^2 + 2 \| \fy - \fc_{i_2}^{\fh} \| \cdot \| \fc_{i_1}^{\fh} - \fc_{i_2}^{\fh} \|
		\nonumber\\&
		\stackrel{(a)}{\leq} \| \fy - \fc_{i_2}^{\fh} \|^2 + \alpha_n^2 + \frac{2\sqrt{a(\sigma_Z^2+\zeta)}}{\bar{n}^{b}},
	\end{align}
	where $(a)$ uses \eqref{Eq.Alpha_nFast} and \eqref{Eq.ui2DistFast}. Thus, for sufficiently large $n,$ we obtain
	\begin{align}
		\label{Ineq.Norm_S_Conv}
		\| \fy - \fc_{i_1}^{\fh} \|^2 - \|\fy - \fc_{i_2}^{\fh} \|^2 \leq \theta .
	\end{align}
	Hence, recalling \eqref{Ineq.Error_Diff} and \eqref{Ineq.Norm_S_Conv} yields
	\begin{align}
		\label{Ineq.GaussianContinuityFast}
		\big| f_{\fZ}(\fy - \fc_{i_1}^{\fh}) - f_{\fZ}(\fy - \fc_{i_2}^{\fh}) \big| & \leq f_{\fZ}(\fy - \fc_{i_1}^{\fh}) \cdot \big| 1 - e^{\frac{\theta}{2\sigma_Z^2}} \big|
		\nonumber\\&
		\leq \tau f_{\fZ}(\fy - \fc_{i_1}^{\fh}),
	\end{align}
	for sufficiently small $\theta>0$ such that $|1-e^{\frac{\theta}{2\sigma_Z^2}}| \leq \tau.$ Now, using \eqref{Eq.ComplTypeIFast} we have the following lower bound on the sum of the type I and type II error probabilities
	\begin{align}
		& P_{e,1}(i_1) + P_{e,2}(i_2,i_1)
		\nonumber\\&
		\geq 1 - \tau - \int_{\mathsf{D}_{i_1,i_2}} f_{\fZ}(\fy - \fc_{i_1}^{\fh})\,d\fy + \int_{\mathsf{D}_{i_1}} f_{\fZ}(\fy - \fc_{i_2}^{\fh}) \,d\fy
		\nonumber\\&
		\geq 1- \tau - \int_{\mathsf{D}_{i_1,i_2}} \big| (f_{\fZ}(\fy - \fc_{i_1}^{\fh}) - f_{\fZ}(\fy - \fc_{i_2}^{\fh})) \big| \,d\fy.
	\end{align}
	Hence, by (\ref{Ineq.GaussianContinuityFast}), we obtain
	\begin{align}
		P_{e,1}(i_1) + P_{e,2}(i_2,i_1) & \geq 1- \tau -\tau \int_{\mathsf{D}_{i_1,i_2}} f_{\fZ}(\fy - \fc_{i_1}^{\fh}) d\fy
		\nonumber\\&
		\geq 1-2\tau,
	\end{align}
	which leads to a contradiction for sufficiently small $\tau$ such that $2\tau < 1 - e_1 - e_2.$ Clearly, this is a contradiction since the error probabilities tend to zero as $n \rightarrow \infty.$ Thus, the assumption in \eqref{Eq.Alpha_nFast} is false. This completes the proof of Lemma~\ref{Lem.Converse}.

	\section*{}
	\bibliographystyle{IEEEtran}
	\bibliography{Lit}
	
\end{document}